\documentclass[preprint,3p,fleqn]{elsarticle}
\usepackage{amsthm}                     
\usepackage{amsmath,amssymb}            
\usepackage{graphicx,threeparttable}    
\usepackage{booktabs,multirow,array}    
\usepackage{xcolor,lineno}              
\usepackage[colorlinks=true]{hyperref}  
\usepackage[linesnumbered,lined,ruled]{algorithm2e}    

\usepackage{pxfonts}
\usepackage[T1]{fontenc}                

\newcommand{\bm}[1]{\boldsymbol{#1}}    

\journal{Applied Mathematical Modelling, 2022, 101: 432-452.}

\newtheorem{theorem}{Theorem}
\newtheorem{lemma}{Lemma}
\newtheorem{definition}{Definition}
\newtheorem{proposition}{Proposition}
\newtheorem{remark}{Remark}

\begin{document}

\begin{frontmatter}

\title{On unified framework for continuous-time grey models: An integral matching perspective}

\author{Baolei Wei}
\ead{weibl@nuaa.edu.cn}

\author{Naiming Xie\corref{cor1}}
\ead{xienaiming@nuaa.edu.cn}
\cortext[cor1]{Corresponding author.}

\address{College of Economics and Management, Nanjing University of Aeronautics and Astronautics, Nanjing 210016, China}

\begin{abstract}
    Since most of the research about grey forecasting models is focused on developing novel models and improving accuracy, relatively limited attention has been paid to the modelling mechanism and relationships among diverse kinds of models. This paper aims to unify and reconstruct continuous-time grey models, highlighting the differences and similarities among different models. First, the unified form of grey forecasting models is proposed and simplified into a reduced-order ordinary differential equation. Then, the integral matching that consists of integral operator and least squares, is proposed to estimate the structural parameter and initial value simultaneously. The cumulative sum operator, an essential element in grey modelling, proves to be the discrete approximation of the integral operator. Next, grey models are reconstructed by the integral matching-based ordinary differential equations. Finally, the existing grey models are compared with the reconstructed models through extensive simulations, and a real-world example shows how to apply and further verify the reconstructed model.
\end{abstract}

\begin{keyword}
    grey system model \sep
    cumulative sum operator \sep
    integral matching \sep
    water supply forecast
\end{keyword}

\end{frontmatter}

\section{Introduction}\label{sec:1}

Over the past four decades since the seminal work on grey system theory \cite{Deng1982Control}, grey system models have been widely used to various fields from natural science through social science. Grey forecasting models which are pioneered in 1984 \cite{deng1984grey}, are constructed based on the prior assumption that the accumulation and release in many generalized energy systems conform to an exponential law \cite{Liu2017Grey}. Grey forecasting models utilize the Cusum (cumulative sum ) operator to mine the exponential characteristic hidden in the original time series and then the continuous-time dynamics (ordinary differential equations) are employed to fit the Cusum series. Generally, grey forecasting models, also refereed as to Grey Models in some literature, are abbreviated as GM($\varphi$, $\psi$), where $\varphi$ denotes the order of derivative and $\phi$ denotes the number of variables. Research has shown consistently that grey forecasting models are promising in solving a class of time series prediction problems, especially in the case of small-sample data sets \cite{Xie2017A}.
In the following, we review grey forecasting models from three viewpoints: the Cusum operator, the basic model and the extended ones.

\subsection{The cumulative sum operator and its inverse}\label{sec:2-1}

\begin{definition}\label{def:01} {(See \cite{li2011extended}.)}
    For the original time series $X{(t)}=\{x(t_1), x(t_2), \cdots, x(t_n)\}$, the Cusum series is defined as $ Y(t)=\left\{y(t_1), y(t_2), \cdots, y(t_n)\right\}$, where
    \[
        y(t_k)=\sum_{i=1}^{k} h_i x(t_i),\quad
        h_k=\begin{cases}
            1, & k=1 \\
            t_k-t_{k-1}, & k\geq 2
        \end{cases}.
    \]
    Correspondingly, the inverse Cusum operator is defined as
    \[
        x(t_k)=\begin{cases}
            \frac{1}{h_1}y(t_1), & k=1 \\
            \frac{1}{h_k}\left[y(t_k)-y(t_{k-1})\right], & k\geq 2
        \end{cases}.
    \]
\end{definition}

Note that the first time interval $h_1$ is set to 1 according to the empty product rule in pure mathematics \cite{ne2009invitation}. The Cusum operator associated with its inverse is an omni-directional formula for not only the equally spaced time series but the irregularly spaced ones \cite{li2011extended,Xiao2014The,yang2019novel}.

Definition \ref{def:01} shows that the original time series can be obtained by restoring the Cusum series. Inspired by this idea, grey forecasting models were first proposed by fitting and forecasting the Cusum series and then restoring the results \cite{deng1984grey}.
In a broad sense, the Cusum operator can be viewed as a data preprocessing technique to non-parametrically mine the pattern hidden in the original time series. There are some other advantages of Cusum operator, especially from a data visualization perspective. For example, it is difficult to identify the hidden pattern from the line graph of the original time series (see Figure \ref{fig:1}(a)), while the quasi-exponential characteristic is obvious in the cumulative sum series (see Figure \ref{fig:1}(b)).

\begin{figure}[!ht]
    \centering
    \includegraphics[scale=0.75, trim = 0 0 0 0, clip = true]{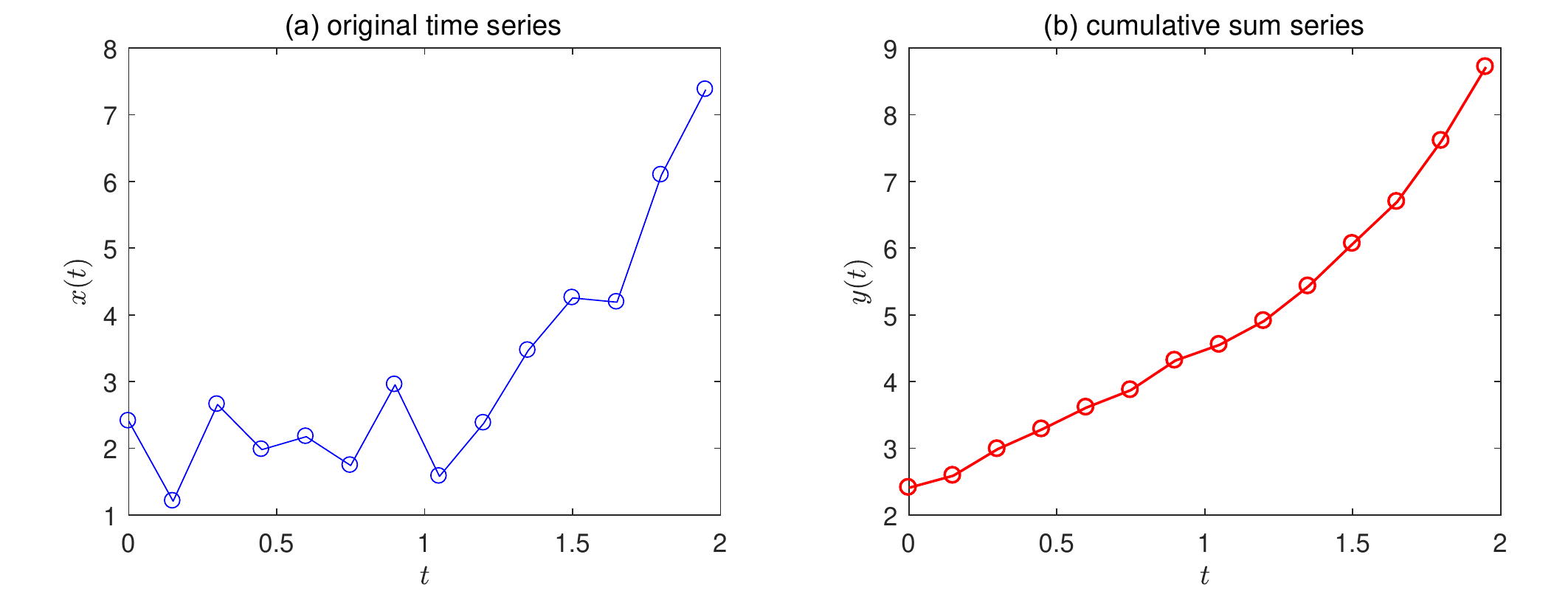}\\
    \caption{Plots of the original time series and the corresponding cumulative sum series.}
    \label{fig:1}
\end{figure}

In fact, grey forecasting models are not the only ones that use the Cusum operator as a pattern recognition tool. The similar idea was also employed in other disciplines, such as ecology \cite{colwell2004interpolating} and statistical process control \cite{chatterjee2009distribution}.

\subsection{The basic grey forecasting model}\label{sec:2-2}

The GM(1,1) model, employing the single-variable and first-order ordinary differential equation, is the most basic and popular grey forecasting model which has been substantiated in many fields. It consists of the continuous-time differential equation
\begin{equation}\label{eq:01}
    \frac{d}{dt} y(t)=a y(t)+b, ~ t \geq t_1
\end{equation}
and the discrete-time difference equation
\begin{equation}\label{eq:02}
    x(t_k)=a \left[\lambda y(t_{k-1})+(1-\lambda) y(t_{k})\right]+b
\end{equation}
where $k=2, 3, \cdots, n$, $n\geq 4$, and $\lambda\in [0,1]$ is a hyperparameter (also referred to as a background coefficient) whose value should be set ahead of the following parameter estimation. By using the least squares, the structural parameters are estimated from equation \eqref{eq:02} as
\begin{equation}\label{eq:03}
    \begin{bmatrix}
        \hat{a} & \hat{b}
    \end{bmatrix}^\mathsf{T}
    =\left(\bm{D}^\mathsf{T}\bm{D}\right)^{-1}\bm{D}^\mathsf{T}\bm{y}
\end{equation}
where
\[
    \bm{D}=\begin{bmatrix}
        \lambda y(t_{1})+(1-\lambda) y(t_{2}) & 1 \\
        \lambda y(t_{2})+(1-\lambda) y(t_{3}) & 1 \\
            \vdots & \vdots \\
        \lambda y(t_{n-1})+(1-\lambda) y(t_{n}) & 1
    \end{bmatrix}, ~
    \bm{y}=\begin{bmatrix}
        x(t_2) \\
        x(t_3) \\
        \vdots \\
        x(t_n)
    \end{bmatrix}.
\]

By setting the initial value parameter to
\begin{equation}\label{eq:04}
    y(t_1)=\eta,
\end{equation}
the time response function is calculated as
\begin{equation}\label{eq:05}
    \hat{y}(t)=\left(\eta+\frac{\hat{b}}{\hat{a}}\right) e^{\hat{a}(t-t_1)}-\frac{\hat{b}}{\hat{a}}, ~ t \geq t_1,
\end{equation}
then by using the inverse Cusum operator, the forecasting results corresponding to the original time series are obtained as $\hat{x}(t_1)=\eta$ and $\hat{x}(t_k)=\frac{1}{h_k}\left[\hat{y}(t_k)-\hat{y}(t_{k-1})\right]$, $k=2,3,\cdots,n+r$, where $r$ is the forecasting horizon.

On the basis of the above modelling process, the existing studies for improving GM(1,1) model can be divided into the following types:
\begin{enumerate}
    \item [(i)]
    Equation \eqref{eq:02} is the numerical discretization-based approximation of equation \eqref{eq:01} by using the weighted trapezoid rule which yields the classical trapezoid formula with the background coefficient equal to 0.5.
    Researchers have been making efforts to improve accuracy by optimizing the background coefficient, such as using  a differential evolution algorithm to search the optimal background coefficient \cite{Zhao2012Using}.
    On the other hand, by using the trapezoid formula in each closed interval $[t_{k-1}, t_{k}]$, the discrete-time equation \eqref{eq:02} is modified to be the following form
    \[
    x(t_k)=a \left[\lambda_k y(t_{k-1})+(1-\lambda_k) y(t_{k})\right]+b
    \]
    where $\lambda_k\in[0,1]$, $k=2,3,\cdots,n$. In each closed interval $[t_{k-1}, t_{k}]$, the background coefficient $\lambda_k$ is adaptively computed by using the triangular membership function rule
    \cite{Li2009An,Li2012Forecasting,chang2015novel} or the exponential function rule \cite{shih2011grey}.
    \item [(ii)]
    The initial value was set to $\hat{y}(t_1)=x(t_1)$ in the classical research \cite{Liu2017Grey}. Therefore, based on the intuition that new observations are likely to be more informative than historical ones, an alternative strategy for the initial value selection is $\hat{y}(t_n)=y(t_n)$ \cite{Dang2002The}. On the other hand, the least-square strategy \cite{Xu2011Improvement} was proposed to further reduce the fitting error:
    \[
    \hat{\eta}=\mathop{\arg\min}_{\eta} \sum_{k=1}^{n} \left( y(t_k)-\hat{y}(t_k) \right)^2
    \]
    where $\hat{y}(t_k)$ is calculated according to the time response equation \eqref{eq:05} which has estimated structural parameters but unknown initial value.
    \item [(iii)]
    Combining the background coefficient optimization and initial value selection together is a straightforward approach to improve modelling accuracy. For example, a hybrid method combined an explicit formula for background coefficient calculation and the least-square strategy for initial value selection \cite{Wang2014Optimization}.
\end{enumerate}

In addition, another research focus of GM(1,1) model is the property analysis, such as the necessary and sufficient condition for modelling \cite{Chen2013The}, the error bound estimation among the variants of GM(1,1) model \cite{Liu2014Error}, the influence of data transformation on modelling accuracy \cite{Tien2009A,Xiao2014The} and the effect of sample size on modelling performance \cite{Yao2009On, Wu2013The}.

\subsection{The extended grey forecasting models}\label{sec:2-3}

Since the procedures of GM(1,1) provide a clear paradigm for grey modelling, \citet{Guo2009Random} probed into the mathematical principles for model extensions but this is short of details. A rough classification of the existing outputs shows that the extended models are mainly focused on the following aspects: (i) extending the linear model into the nonlinear ones; (ii) extending the signal-variable model into the multi-variable and multi-output ones.

First, the extended single-variable models have a similar modelling procedures with GM(1,1) model, i.e., modelling the Cusum series with ordinary differential equations. The main differences lie in the continuous-time and discrete-time equations. Table \ref{tbl:1} shows the linear extensions have similar expressions, excluding the forcing terms in differential equations.

\begin{table}[!htb]
    \centering
    \begin{threeparttable}
    \caption{The coupled equations of the extended single-variable linear and nonlinear grey forecasting models.}
    \label{tbl:1}
    \begin{tabular}{l l l l l}
    \toprule
       Type  & Model & Differential equation & Difference equation & Ref \\ \hline
            Linear  & GM(1,1) & $\dfrac{d}{dt} y(t)=a y(t)+b$
                        & $x(k)=a \frac{y(k-1)+y(k)}{2}+b$ & \cite{Liu2017Grey} \\
                   & NGM(1,1,$k$) & $\dfrac{d}{dt} y(t)=a y(t)+b{t}$
                        & $x(k)=a \frac{y(k-1)+y(k)}{2}+bk$ & \cite{Cui2013A} \\
                   & NGM(1,1,$k$,$c$) & $\dfrac{d}{dt} y(t)=a y(t)+b_1{t}+b_0$
                        & $x(k)=a \frac{y(k-1)+y(k)}{2}+b_1 k+b_0$ & \cite{chen2014foundation}\\
                   & GM(1,1,$t^\alpha$) & $\dfrac{d}{dt} y(t)=a y(t)+b_1{t^\alpha}+b_0$
                        & $x(k)=a \frac{y(k-1)+y(k)}{2}+b_1 k^\alpha+b_0$ & \cite{ding2021application} \\ 
                   & GPM(1,1,$N$) & $\dfrac{d}{dt} y(t)=a y(t)+\sum\limits_{j=0}^{N}b_j{t^j}$
                        & $x(k)=a \frac{y(k-1)+y(k)}{2}+\sum\limits_{j=0}^{N}b_j\frac{k^{j+1}-(k-1)^{j}}{j+1}$ & \cite{Luo2017Grey} \\
                   & KRNGM(1,1) & $\dfrac{d}{dt} y(t)=a y(t)+\sum\limits_{j=0}^{N}b_j\phi_j(t)$
                        &  $x(k)=a \frac{y(k-1)+y(k)}{2}+\sum\limits_{j=0}^{N}b_j\frac{\phi_j(k)+\phi_j(k-1)}{2}$ & \cite{Ma2016A} \\
           \multirow{1}*{\shortstack{Nonlinear}}
                & Verhulst & $\dfrac{d}{dt} y(t)=a y(t)+b\left[ y(t) \right]^2$
                        & $x(k)=a\frac{y(k-1)+y(k)}{2}+b\left(\frac{y(k-1)+y(k)}{2}\right)^2$ & \cite{evans2014alternative} \\
                & NBGM(1,1) & $\dfrac{d}{dt} y(t)=a y(t)+b\left[ y(t) \right]^\gamma$
                        & $x(k)=a\frac{y(k-1)+y(k)}{2}+b\left(\frac{y(k-1)+y(k)}{2}\right)^\gamma$ & \cite{chen2008forecasting} \\
    \bottomrule
    \end{tabular}
    \begin{tablenotes}
        \item [~] Note that the time interval in these models is set to unit, i.e., $t_k=k$, $k=1,2,\cdots,n$.
    \end{tablenotes}
    \end{threeparttable}
\end{table}

Table \ref{tbl:1} shows that the extensions have an inclusion relation, i.e., the former is a special case of the latter in the linear and nonlinear cases, respectively. In particular, both GPM(1,1,$N$) and KRNGM(1,1) are extensions from a basis expansion perspective; the former uses the polynomial basis function, while the latter utilizes the kernel function. The Grey Verhulst model, aimed at fitting the inverted U-shaped series, is the first nonlinear extension, and then the power coefficient is generalized from the fixed integer 2 to a real number, in order to improve the flexibility in NBGM(1,1) model \cite{yang2021integral}.

Second, the multi-variable extensions corresponding to the single-variable models were proposed successively. The first multi-variable extension GM(1,$N$) plays an important role in establishing the multi-variable models, even though it is an incomplete extension of GM(1,1)  \cite{Liu2017Grey}; then on this basis, \citet{tien2012research} came to the complete extension by introducing a control parameter. Subsequently, researchers proposed some other multi-variable models including the NGM(1,$N$) \cite{zeng2016development} which is the extension of NGM(1,1,$k$,$c$), the KGM(1,$N$) \cite{ma2018kernel} which is the extension of KRNGM(1,1), and the NBGM(1,$N$) \cite{wang2017forecasting,yu2021novel} which is the extension of NBGM(1,1) model.

Third, further extension of multi-variable models leads to the multi-output ones, but the literature is generally sparse on this up to now. To the best of our knowledge, these mainly include the linear MGM(1,$N$) model \cite{xiong2017mgm} associated with its improvement coupled with self-memory \cite{guo2015multi}, and the nonlinear MNBGM(1,$N$) model coupled with self-memory \cite{guo2019prediction}, as well as the nonlinear grey Lotka--Volterra model \cite{wu2012grey,wang2016application}.

Finally, in a similar manner to the research route of the GM(1,1) model, there have been a variety of optimization studies on the aforementioned three categories of extension models, such as the background coefficient and/or initial value optimization for (i) the single-variable cases such as the GPM(1,1,$N$) model \cite{wei2018optimal} and the NBGM(1,1) model \cite{ding2021novel} and (ii) the multi-variable cases such as the NGM(1,$N$) model \cite{Zeng2018Improved}.

\subsection{Motivation}

Having reviewed a wide range of grey forecasting models in the previous introductory section, the next sections will address, in detail, certain challenging issues. The main objectives of this paper are as follows:
\begin{enumerate}
    \item [(i)]
    There exist a variety of grey forecasting models with different model representations (see Table \ref{tbl:1} and the multi-variable and multi-output models that follow), making it difficult for researchers to perform property analysis; and for practitioners to select an appropriate model for a given practical problem.
    \item [(ii)]
    The Cusum operator is a fundamental element in the grey forecasting model but, until now, the mechanism has not been explained clearly, especially from a mathematical perspective.
    \item [(iii)]
    Forecasting the future values of the original time series is the main goal of a grey model, so is it possible to model the original time series directly but achieve the same performance without explicit use of Cusum operator? In this way, the modelling procedures will be simplified because it avoids the inverse Cusum operator when calculating the forecasts of the original time series, thus making the modelling results easier to explain.
\end{enumerate}

The principal contributions of the present research are as follows:
\begin{enumerate}
    \item [(i)]
    The aforementioned linear grey forecasting models, together with their multi-variable and multi-output extensions, are unified by a matrix differential equation, after extending the Cusum operator for single-variable time series to one that is suitable for vector time series.
    \item [(ii)]
    By introducing an integral operator, the Cusum operator proves to be the Euler's formula-based numerical discretization of this operator; and, subsequently, the mechanism of the Cusum operator is further explained within the parameter estimation process.
    \item [(iii)]
    The unified grey forecasting model for Cusum series proves to be equivalent to a reduced-order ordinary differential equation for the original series. This provides the basis for reconstructing grey forecasting models using an integral matching-based ordinary differential equation framework, in which the differential equation is first converted to an integral equation; and then the structural parameter and initial value are estimated simultaneously.
\end{enumerate}

The remaining parts of the paper are organized as follows:
in section \ref{sec:3}, we extend the Cusum operator, propose a unified representation, reconstruct the modelling process using matrix analysis, and discuss some special cases of the unified model;
in section \ref{sec:4}, we introduce the integral operator to explain the Cusum mechanism, obtain the equivalent reduced-order differential equation, and present integral matching as a method for simultaneously estimating the structural parameter and initial value;
in section \ref{sec:a-1}, we analyze the relationship between the grey model and the integral matching-based one, then reconstruct the modelling procedures of the grey model;
section \ref{sec:5} provides large-scale simulations that compare the two approaches from both parameter estimation and forecasting accuracy viewpoints;
section \ref{sec:6} provides a real-world example to show how to apply our method; while section \ref{sec:7} concludes the work and discusses the future directions of research and development.

\section{Unified representation of grey forecasting models}\label{sec:3}

In this section we present a universal framework for the existing linear grey forecasting models, including structural parameter estimation and initial value selection strategies.

For a state vector $\bm{x}(t)\in \mathbb{R}^d$, the observational data sampled at time instants $t_1$, $t_2$, $\cdots$, $t_n$ are arranged into the following matrix form:
\[
   \begin{bmatrix}
        \bm{x}^\mathsf{T}(t_1) \\
        \bm{x}^\mathsf{T}(t_2) \\
        \vdots \\
        \bm{x}^\mathsf{T}(t_n)
    \end{bmatrix}
    =
   \begin{bmatrix}
        x_1(t_1) & x_2(t_1) & \cdots & x_d(t_1) \\
        x_1(t_2) & x_2(t_2) & \cdots & x_d(t_2) \\
        \vdots & \vdots & \vdots & \vdots \\
        x_1(t_n) & x_2(t_n) & \cdots & x_d(t_n)
    \end{bmatrix}
\]
and subsequently the corresponding cumulative sum form is expressed as
\[
    \begin{bmatrix}
        \bm{y}^\mathsf{T}(t_1) \\
        \bm{y}^\mathsf{T}(t_2) \\
        \vdots \\
        \bm{y}^\mathsf{T}(t_n)
    \end{bmatrix}
    =
    \begin{bmatrix}
        y_1(t_1) & y_2(t_1) & \cdots & y_d(t_1) \\
        y_1(t_2) & y_2(t_2) & \cdots & y_d(t_2) \\
        \vdots & \vdots & \vdots & \vdots \\
        y_1(t_n) & y_2(t_n) & \cdots & y_d(t_n)
    \end{bmatrix}
\]
where the Cusum column vector is calculated according to the formula  $\bm{y}(t_k)=\sum_{i=1}^{k} h_i \bm{x}(t_i)$ which is the multi-variable extension of the single-variable Cusum operator in Definition \ref{def:01}.

For the multi-variable Cusum series, consider the grey forecasting models involving the ordinary differential equation
\begin{equation}\label{eq:gm}
    \frac{d}{dt} \bm{y}(t)=\bm{A}\bm{y}(t)+\bm{B}\bm{u}(t)+\bm{c}, ~ t \geq t_1
\end{equation}
where $\bm{y}(t) \in \mathbb{R}^d$ is the Cusum vector, $\bm{u}(t)\in\mathbb{R}^p$ is a known vector (that is independent of $\bm{y}(t)$), $\bm{A}\in \mathbb{R}^{d \times d}$, $\bm{B}\in \mathbb{R}^{d \times p}$ and $\bm{c}\in \mathbb{R}^{d}$ are unknown structural parameters.

First, numerical discretization-based gradient matching is employed to estimate the structural parameters in equation \eqref{eq:gm}. This method can be divided into the following two steps: in the first step, using the trapezoid formula gives the discrete-time equation
\begin{equation}\label{eq:08}
    \frac{\bm{y}(t_k)-\bm{y}(t_{k-1})}{t_k-t_{k-1}}
    =\bm{A}\frac{\bm{y}(t_{k-1})+\bm{y}(t_{k})}{2}+
        \bm{B}\frac{\bm{u}(t_{k-1})+\bm{u}(t_{k})}{2}+\bm{c}+\mathcal{O}(h^2)
\end{equation}
where $h=\max\nolimits_{k=2}^{n}\{h_k\}$.
Realistically, often only $\bm{x}(t)$ is available and contaminated with noise. Thus, $\bm{y}(t)$ is contaminated with noise and equation \eqref{eq:08} does not hold exactly. By substituting $\bm{x}(t_k)$ for $\frac{\bm{y}(t_k)-\bm{y}(t_{k-1})}{t_k-t_{k-1}}$, the expression \eqref{eq:08} can be rewritten as
\begin{equation}\label{eq:09}
    \bm{x}(t_k)=\bm{A}\frac{\bm{y}(t_{k-1})+\bm{y}(t_{k})}{2}+
        \bm{B}\frac{\bm{u}(t_{k-1})+\bm{u}(t_{k})}{2}+\bm{c}+\bm{\varepsilon}(t_k)
\end{equation}
where $\bm{\varepsilon}(t_k)$ is the sum of discretization error and noise error. By substituting $k=2,3,\cdots,n$ into equation \eqref{eq:09} and arranging the resulting $n-1$ equations into a matrix form, one has
\begin{equation}\label{eq:10}
    \bm{X}=\bm{\Theta}(\bm{y}, \bm{u}) \bm{\Xi}+\bm{\Gamma}
\end{equation}
where
\[
    \bm{\Xi}=
        \begin{bmatrix}
            \bm{A}^\mathsf{T} \\
            \bm{B}^\mathsf{T} \\
            \bm{c}^\mathsf{T}
        \end{bmatrix}, ~
    \bm{X}=
        \begin{bmatrix}
            \bm{x}^\mathsf{T}(t_2) \\
            \bm{x}^\mathsf{T}(t_3) \\
            \vdots \\
            \bm{x}^\mathsf{T}(t_n)
        \end{bmatrix},~
    \bm{\Theta}(\bm{y}, \bm{u})=
        \begin{bmatrix}
            \frac{\bm{y}^\mathsf{T}(t_1)+\bm{y}^\mathsf{T}(t_2)}{2} & \frac{\bm{u}^\mathsf{T}(t_1)+\bm{u}^\mathsf{T}(t_2)}{2} & 1 \\
            \frac{\bm{y}^\mathsf{T}(t_2)+\bm{y}^\mathsf{T}(t_3)}{2} & \frac{\bm{u}^\mathsf{T}(t_2)+\bm{u}^\mathsf{T}(t_3)}{2} & 1 \\
            \vdots & \vdots & \vdots \\
            \frac{\bm{y}^\mathsf{T}(t_{n-1})+\bm{y}^\mathsf{T}(t_n)}{2} & \frac{\bm{u}^\mathsf{T}(t_{n-1})+\bm{u}^\mathsf{T}(t_n)}{2} & 1
        \end{bmatrix},~
    \bm{\Gamma}=
        \begin{bmatrix}
            \bm{\varepsilon}^\mathsf{T}(t_2) \\
            \bm{\varepsilon}^\mathsf{T}(t_3) \\
            \vdots \\
            \bm{\varepsilon}^\mathsf{T}(t_n)
        \end{bmatrix}.
\]

Then, in a second step, we obtain the numerical discretization-based estimate $\hat{\bm{\Xi}}$ of $\bm{\Xi}$ by minimizing the least-squares criterion
\begin{equation}\label{eq:11}
    \mathcal{L}(\bm{\Xi})
    =\left\|\bm{X}-\bm{\Theta}(\bm{y}, \bm{u}) \bm{\Xi}\right\|_\mathsf{F}^2
\end{equation}
where $\left\|\cdot\right\|_\mathsf{F}$ is the Frobenius norm.
The condition for this is obtained in the usual manner by partially differentiating the objective function $\mathcal{L}(\bm{\Xi})$ with respect to the parameter matrix, in turn, and setting the derivative to zero. This yields the matrix equation
\begin{align*}
    \frac{\partial}{\partial \bm{\Xi}} \mathcal{L}(\bm{\Xi})
        = \frac{\partial}{\partial \bm{\Xi}} \mathsf{Tr}
            \left(\left[\bm{X}-\bm{\Theta}(\bm{y}, \bm{u}) \bm{\Xi}\right]^\mathsf{T} \left[\bm{X}-\bm{\Theta}(\bm{y}, \bm{u}) \bm{\Xi}\right] \right)
        = 2\bm{\Theta}^\mathsf{T}(\bm{y},\bm{u})\bm{\Theta}(\bm{y},\bm{u})\bm{\Xi}-
            2\bm{\Theta}^\mathsf{T}(\bm{y},\bm{u})\bm{X}
        = 0
\end{align*}
and the subsequent least-square estimate
\begin{equation}\label{eq:12}
    \hat{\bm{\Xi}}
    =\left(\bm{\Theta}^\mathsf{T}(\bm{y},\bm{u})\bm{\Theta}(\bm{y},\bm{u})\right)^{-1}
    \bm{\Theta}^\mathsf{T}(\bm{y},\bm{u})\bm{X}.
\end{equation}

Next, we solve equation \eqref{eq:gm} using the estimated parameters. By employing the variation of parameters method, one has the general solution (also called time response function) to the non-homogeneous equation \eqref{eq:gm} expressed as
\begin{equation}\label{eq:13}
    \breve{\bm{y}}(t)=
        \exp\left(\hat{\bm{A}}(t-t_1)\right)
        \left\{\bm{\eta}+\int_{t_1}^{t}\exp\left(\hat{\bm{A}}(t_1-s)\right)\left[\hat{\bm{B}}\bm{u}(s)+\hat{\bm{c}}\right]ds\right\}
\end{equation}
where $\breve{\bm{y}}(t_1)=\bm{\eta} \in \mathbb{R}^d$ is the unknown initial vector. In order to forecast the evolutions, the required initial vector $\bm{\eta}$ is always determined by using one of the following strategies:
\begin{enumerate}
    \item [(i)]
        the fixed first point strategy assuming
            \[
                \breve{\bm{y}}(t)=\bm{y}(t_1)
                \quad \Rightarrow \quad
                \bm{\eta}=\bm{y}(t_1),
            \]
        which leads to the time response function
        \begin{equation}\label{eq:14}
            \hat{\bm{y}}(t)=
            \exp\left(\hat{\bm{A}}(t-t_1)\right)
            \left\{\bm{y}(t_1)+\int_{t_1}^{t}\exp\left(\hat{\bm{A}}(t_1-s)\right)\left[\hat{\bm{B}}\bm{u}(s)+\hat{\bm{c}}\right]ds\right\};
        \end{equation}
    \item [(ii)]
        the fixed last point strategy assuming
            \[
                \breve{\bm{y}}(t_n)=\bm{y}(t_n)
                \quad \Rightarrow \quad
                \bm{\eta}=\exp\left(\hat{\bm{A}}(t_1-t_n)\right)\bm{y}(t_n)-\int_{t_1}^{t_n}\exp\left(\hat{\bm{A}}(t_1-s)\right)\left[\hat{\bm{B}}\bm{u}(s)+\hat{\bm{c}}\right]ds,
            \]
        which leads to the time response function
        \begin{equation}\label{eq:15}
            \hat{\bm{y}}(t)=
            \exp\left(\hat{\bm{A}}(t-t_1)\right)
            \left\{
                \exp\left(\hat{\bm{A}}(t_1-t_n)\right)\bm{y}(t_n) + \int_{t_n}^{t}\exp\left(\hat{\bm{A}}(t_1-s)\right)\left[\hat{\bm{B}}\bm{u}(s)+\hat{\bm{c}}\right]ds
            \right\};
        \end{equation}
    \item [(iii)]
        the least-squares strategy minimizing the least-square objective function
        \[
            \breve{\bm{y}}(t_1) =
            \hat{\bm{\eta}} = \mathop{\arg\min}_{\bm{\eta}}
            \left\{
                \mathcal{L}(\bm{\eta})=\sum_{k=1}^{n}\left\| \bm{y}(t_k)-\breve{\bm{y}}(t_k) \right\|_2^2
            \right\},
        \]
        which subsequently leads to the time response function
        \begin{equation}\label{eq:16}
            \hat{\bm{y}}(t)=
            \exp\left(\hat{\bm{A}}(t-t_1)\right)
            \left\{\hat{\bm{\eta}}+\int_{t_1}^{t}\exp\left(\hat{\bm{A}}(t_1-s)\right)\left[\hat{\bm{B}}\bm{u}(s)+\hat{\bm{c}}\right]ds\right\}.
        \end{equation}
\end{enumerate}

Equations \eqref{eq:14}, \eqref{eq:15} and \eqref{eq:16} show the most commonly-used strategies to determine the initial vector. Finally, substituting the time points $\left\{t_k\right\}_{k=1}^{n+r}$ into equation \eqref{eq:14}, \eqref{eq:15} or \eqref{eq:16} yields the fitting and forecasting values of the Cusum series and then, using the inverse Cusum operator, we obtain the fitting and forecasting results of the original multi-variable time series:
\begin{equation}\label{eq:06}
    \hat{\bm{x}}(t_k)=\begin{cases}
        \hat{\bm{y}}(t_1), & k=1 \\
        \frac{1}{t_k-t_{k-1}}\left[ \hat{\bm{y}}(t_k)-\hat{\bm{y}}(t_{k-1}) \right], & k=2,3,\cdots,n+r \\
    \end{cases}.
\end{equation}

Overall, the modelling procedures of grey forecasting models can be summarized easily from the above process. Then, we are able to analyze this unified model from three different scenarios including the single-variable, multi-variable and multi-output ones.

\begin{remark}\label{rem:1}
    Let the dimension of the state vector be $d=1$. If $\bm{u}(t)=\begin{bmatrix} u_\iota(t) \end{bmatrix}_{p\times1}$ consists of the fixed basis functions of time, then the unified model yields multiple families of single-variable grey forecasting models.
    \begin{enumerate}
        \item [(1)]
             For the family of grey polynomial models, the basis functions are $u_\iota(t)=t^\iota$, $\iota=1,2,\cdots,p$. By inserting $u_\iota(t)$ back into equations \eqref{eq:gm} and \eqref{eq:09}, we obtain the continuous- and discrete-time equations respectively expressed as
             \[
                \frac{d}{dt}y(t)=ay(t)+\sum\limits_{\iota=1}^{p}{b_\iota t^\iota}+c
             \]
             and
             \[
                x(t_k)=a\frac{y(t_{k-1})+y(t_k)}{2}+\sum\limits_{\iota=1}^{p}{b_\iota \frac{t_{k-1}^\iota+t_{k}^\iota}{2}}+c+\varepsilon(t_k).
             \]
        \item [(2)]
            For the family of grey Fourier models, the basis functions are
            $u_{2\iota-1}(t)=\sin\left(2\iota\pi ft\right)$ and  $u_{2\iota}(t)=\cos\left({2\iota\pi}ft\right)$, $\iota=1,2,\cdots,\frac{p}{2}$, where $p$ is even and $f$ is frequency. By inserting $u_{2\iota-1}(t)$ and $u_{2\iota}(t)$ back into equations \eqref{eq:gm} and \eqref{eq:09}, we obtain the continuous- and discrete-time equations respectively expressed as
            \[
                \frac{d}{dt}y(t)=ay(t)+
                    \sum\limits_{\iota=1}^{p/2} {b_{2\iota-1}\sin\left(2{\iota}\pi f t\right)}
                   +\sum\limits_{\iota=1}^{p/2} {b_{2\iota}\cos\left(2{\iota}\pi f t\right)}+c
            \]
            and
            \begin{multline*}
            x(t_k)=a\frac{y(t_{k-1})+y(t_k)}{2}+
                \sum\limits_{\iota=1}^{p/2} b_{2\iota-1}
                        \frac{\sin\left(2{\iota}\pi f t_{k-1}\right)+\sin\left(2{\iota}\pi f t_{k}\right)}{2} \\
               +\sum\limits_{\iota=1}^{p/2} b_{2\iota}
                        \frac{\cos\left(2{\iota}\pi f t_{k-1}\right)+\cos\left(2{\iota}\pi f t_{k}\right)}{2}
                 +c+\varepsilon(t_k).
            \end{multline*}
    \end{enumerate}
\end{remark}

Remark \ref{rem:1} shows that this unified model subsumes a number of special families of grey models, although here only two of them are exemplified. Furthermore, each family also subsumes many specific models that have been researched and received specific names. For instance,  all the linear models in Table \ref{tbl:1} belong to the grey polynomial model family and, additionally, the grey Fourier model family \cite{comert2021improved} covers the hybrid grey models which uses the Fourier series fitting to modify the residual errors of GM(1,1) model \cite{lin2007novel, lin2009adaptive, kayacan2010grey, xiong2014optimal}.

\begin{remark}\label{rem:2}
    Let the dimension of the state vector be $d=1$. If the vector $\bm{u}(t)=\begin{bmatrix} u_\iota(t) \end{bmatrix}_{p\times1}$ consists of other state variables, i.e., $u_\iota(t)=y_\iota(t)$, $\iota=1, 2, \cdots, p$, then the unified model yields the multi-variable grey forecasting model \cite{tien2012research,kung2008prediction}, whose continuous- and discrete-time equations are
    \[
        \frac{d}{dt}y(t)=ay(t)+\sum\limits_{\iota=1}^{p}{b_\iota y_\iota(t)}+c
    \]
    and
    \[
        x(t_k)=a\frac{y(t_{k-1})+y(t_k)}{2}+
            \sum\limits_{\iota=1}^{p}{b_\iota \frac{y_\iota(t_{k-1})+y_\iota(t_{k})}{2}}+c+\varepsilon(t_k).
    \]
\end{remark}

If the vector $\bm{u}(t)=\begin{bmatrix} u_\iota(t) \end{bmatrix}_{p\times1}$ consists of a mixture of other state variables and known fixed basis functions of time, then combining the extension principle in Remark \ref{rem:1} gives
\[
    \frac{d}{dt}y(t)=ay(t)+
        \sum\limits_{\iota=1}^{p_1}{b_\iota y_\iota(t)}+
        \sum\limits_{\iota=1}^{p_2}{b_\iota u_\iota(t)}+c
\]
and
\[
    x(t_k)=a\frac{y(t_{k-1})+y(t_k)}{2}+
        \sum\limits_{\iota=1}^{p_1}{b_\iota \frac{y_\iota(t_{k-1})+y_\iota(t_{k})}{2}}+
        \sum\limits_{\iota=1}^{p_2}{b_\iota \frac{u_\iota(t_{k-1})+u_\iota(t_{k})}{2}}+ c+\varepsilon(t_k),
\]
which shows that this model yields many multi-variable grey forecasting models, such as the aforementioned NGM(1,$N$) model \cite{zeng2016development} ($p_1=N-1$, $p_2=1$ and $u_1(t)=t$). This model may also lead to many other novel multi-variable grey forecasting models.

\begin{remark}\label{rem:3}
    Let the dimension of the state vector be $d\geq 2$. If the vector $\bm{u}(t)=\begin{bmatrix} u_\iota(t) \end{bmatrix}_{p\times1}$ is equal to $\bm{0}$, i.e., $u_\iota(t)=0$, $\iota=1, 2, \cdots, p$, then the unified model yields the multi-output grey forecasting model \cite{xiong2017mgm,xiong2020examination}, whose continuous- and discrete-time equations are
    \[
    \begin{cases}
        \frac{d}{dt} y_1(t) = \sum\limits_{\iota=1}^{d} a_{1,\iota}y_{\iota}(t)+c_{1} \\
        \frac{d}{dt} y_2(t) = \sum\limits_{\iota=1}^{d} a_{2,\iota}y_{\iota}(t)+c_{2} \\
        \quad \quad \quad ~ \vdots \\
        \frac{d}{dt} y_d(t) = \sum\limits_{\iota=1}^{d} a_{d,\iota}y_{\iota}(t)+c_{d} \\
    \end{cases}
    \]
    and
    \[
    \begin{cases}
        x_1(t_k) = \sum\limits_{\iota=1}^{d} a_{1,\iota} \frac{y_{\iota}(t_{k-1})+y_{\iota}(t_{k})}{2}+c_{1}+\varepsilon_1(t_k) \\
        x_2(t_k) = \sum\limits_{\iota=1}^{d} a_{2,\iota} \frac{y_{\iota}(t_{k-1})+y_{\iota}(t_{k})}{2}+c_{2}+\varepsilon_2(t_k) \\
        \quad \quad \quad ~ \vdots \\
        x_d(t_k) = \sum\limits_{\iota=1}^{d} a_{d,\iota} \frac{y_{\iota}(t_{k-1})+y_{\iota}(t_{k})}{2}+c_{d}+\varepsilon_d(t_k) \\
    \end{cases}.
    \]
\end{remark}

Similar to the extensions in Remarks \ref{rem:1} and \ref{rem:2}, if the vector $\bm{u}(t)=\begin{bmatrix} u_\iota(t) \end{bmatrix}_{p\times1}$ consists of fixed functions of time, then the continuous- and discrete-time equations are expressed as
    \[
    \begin{cases}
        \frac{d}{dt} y_1(t) = \sum\limits_{\iota=1}^{d} a_{1,\iota}y_{\iota}(t)+\sum\limits_{\iota=1}^{p}b_{1,\iota}u_{\iota}(t)+c_{1} \\
        \frac{d}{dt} y_2(t) = \sum\limits_{\iota=1}^{d} a_{2,\iota}y_{\iota}(t)+\sum\limits_{\iota=1}^{p}b_{2,\iota}u_{\iota}(t)+c_{2} \\
        \quad \quad \quad ~ \vdots \\
        \frac{d}{dt} y_d(t) = \sum\limits_{\iota=1}^{d} a_{d,\iota}y_{\iota}(t)+\sum\limits_{\iota=1}^{p}b_{d,\iota}u_{\iota}(t)+c_{d} \\
    \end{cases}
    \]
    and
    \[
    \begin{cases}
        x_1(t_k) = \sum\limits_{\iota=1}^{d} a_{1,\iota} \frac{y_{\iota}(t_{k-1})+y_{\iota}(t_{k})}{2}+
        \sum\limits_{\iota=1}^{p} b_{1,\iota}\frac{u_{\iota}(t_{k-1})+u_{\iota}(t_{k})}{2}+c_{1}+\varepsilon_1(t_k) \\
        x_2(t_k) = \sum\limits_{\iota=1}^{d} a_{2,\iota} \frac{y_{\iota}(t_{k-1})+y_{\iota}(t_{k})}{2}+
        \sum\limits_{\iota=1}^{p} b_{2,\iota}\frac{u_{\iota}(t_{k-1})+u_{\iota}(t_{k})}{2}+c_{2}+\varepsilon_2(t_k) \\
        \quad \quad \quad ~ \vdots \\
        x_d(t_k) = \sum\limits_{\iota=1}^{d} a_{d,\iota} \frac{y_{\iota}(t_{k-1})+y_{\iota}(t_{k})}{2}+
        \sum\limits_{\iota=1}^{p} b_{d,\iota}\frac{u_{\iota}(t_{k-1})+u_{\iota}(t_{k})}{2}+c_{d}+\varepsilon_d(t_k) \\
    \end{cases}
    \]
which can not only yield the existing multi-output models, such as the ones coupling self-memory theory \cite{guo2015multi} and multiple regression \cite{xiong2011combined}, but suggest other novel ones.

Remarks \ref{rem:1}--\ref{rem:3} show that the unified equation \eqref{eq:gm} has the capacity to represent the continuous-time single-variable, multi-variable and multi-output grey models; and furthermore, it can also introduce the possibility of developing some other novel models.

\section{Integral matching-based ordinary differential equation models}\label{sec:4}

In this section we deduce the reduced-order representation of the aforementioned unified form; then the integral matching, which is composed of integral operator and least squares  \cite{dattner2015optimal,dattner2015model,wei2019understanding}, is introduced to explain the mechanism of the cumulative sum operator as well as estimating the structural parameters and initial values simultaneously.

\subsection{Reduced-order representation of grey forecasting models}

\begin{lemma}\label{lem:01}
    Let the translation transformation of the Cusum series be $\bm{y}_\mathrm{tr}(t_k)=\bm{y}(t_k)+\bm{\xi}$, $k=1,2,\cdots,n$, where $\bm{\xi}\in \mathbb{R}^d$ is a real vector. Then,
        (i) the estimated parameters satisfy $\hat{\bm{A}}_\mathrm{tr}=\hat{\bm{A}}$, $\hat{\bm{B}}_\mathrm{tr}=\hat{\bm{B}}$, $\hat{\bm{c}}_\mathrm{tr}=\hat{\bm{c}}-\hat{\bm{A}}\bm{\xi}$;
        (ii) the time response functions satisfy $\hat{\bm{y}}_\mathrm{tr}(t)=\hat{\bm{y}}(t)+\bm{\xi}$;
        (iii) the fitting and forecasting values of the original time series satisfy $\hat{\bm{x}}_\mathrm{tr}(t_k)=\hat{\bm{x}}(t_k)$, $k=2,3,\cdots,n$.
\end{lemma}

\begin{proof}
    By substituting $\bm{y}_\mathrm{tr}(t_k)=\bm{y}(t_k)+\bm{\xi}$ into the matrices in equation \eqref{eq:10}, the two matrices therein are rewritten as
    \[
        \bm{X}_\mathrm{tr}=\bm{X}
        \text{~ with ~}
        \bm{x}_\mathrm{tr}(t_k)=\bm{x}(t_k)=\frac{1}{h_k}\left[ \bm{y}_\mathrm{tr}(t_k)-\bm{y}_\mathrm{tr}(t_{k-1}) \right]$,~ $k=2,3,\cdots,n
    \]
    and
    \[
    \begin{split}
        \bm{\Theta}(\bm{y}_\mathrm{tr}, \bm{u})
         =  \begin{bmatrix}
                \frac{\bm{y}_\mathrm{tr}^\mathsf{T}(t_1)+\bm{y}_\mathrm{tr}^\mathsf{T}(t_2)}{2} & \frac{\bm{u}^\mathsf{T}(t_1)+\bm{u}^\mathsf{T}(t_2)}{2} & 1 \\
                \frac{\bm{y}_\mathrm{tr}^\mathsf{T}(t_2)+\bm{y}_\mathrm{tr}^\mathsf{T}(t_3)}{2} & \frac{\bm{u}^\mathsf{T}(t_2)+\bm{u}^\mathsf{T}(t_3)}{2} & 1 \\
                \vdots & \vdots & \vdots \\
                \frac{\bm{y}_\mathrm{tr}^\mathsf{T}(t_{n-1})+\bm{y}_\mathrm{tr}^\mathsf{T}(t_n)}{2} & \frac{\bm{u}^\mathsf{T}(t_{n-1})+\bm{u}^\mathsf{T}(t_n)}{2} & 1
            \end{bmatrix}
        =  \bm{\Theta}(\bm{y}, \bm{u})
            \begin{bmatrix}
                \bm{I}_d & \bm{0} & \bm{0} \\
                \bm{0} & \bm{I}_p & \bm{0} \\
                \bm{\xi}^\mathsf{T} & \bm{0} & 1
            \end{bmatrix}
    \end{split}.
    \]

    From equation \eqref{eq:12}, one has the estimated parameters expressed as
    \[
    \begin{bmatrix}
        \hat{\bm{A}}^\mathsf{T}_\mathrm{tr} \\
        \hat{\bm{B}}^\mathsf{T}_\mathrm{tr} \\
        \hat{\bm{c}}^\mathsf{T}_\mathrm{tr}
    \end{bmatrix}
    =\left(\bm{\Theta}^\mathsf{T}(\bm{y}_\mathrm{tr},\bm{u})\bm{\Theta}(\bm{y}_\mathrm{tr},\bm{u})\right)^{-1}
        \bm{\Theta}^\mathsf{T}(\bm{y}_\mathrm{tr},\bm{u})\bm{X}_\mathrm{tr}
    =\begin{bmatrix}
        \bm{I}_d & \bm{0} & \bm{0} \\
        \bm{0} & \bm{I}_p & \bm{0} \\
        \bm{\xi}^\mathsf{T} & \bm{0} & 1 \\
      \end{bmatrix}^\mathrm{-1}
      \begin{bmatrix}
          \hat{\bm{A}}^\mathsf{T} \\
          \hat{\bm{B}}^\mathsf{T} \\
          \hat{\bm{c}}^\mathsf{T}
      \end{bmatrix}
     =\begin{bmatrix}
          \hat{\bm{A}}^\mathsf{T} \\
          \hat{\bm{B}}^\mathsf{T} \\
          \hat{\bm{c}}^\mathsf{T}-\bm{\xi}^\mathsf{T}\hat{\bm{A}}^\mathsf{T}
      \end{bmatrix}.
    \]

    Therefore, the corresponding time response function obtained from equation \eqref{eq:13} is
    \[
    \begin{split}
        \hat{\bm{y}}_\mathrm{tr}(t)
         &= \exp\left(\hat{\bm{A}}(t-t_1)\right)
            \left\{ \bm{\eta}_\mathrm{tr} +\int_{t_1}^{t}
              \exp\left(\hat{\bm{A}}(t_1-s)\right)
              \left[ \hat{\bm{B}}\bm{u}(s)+\hat{\bm{c}} \right]ds - \int_{t_1}^{t}\exp\left(\hat{\bm{A}}(t_1-s)\right)\hat{\bm{A}}\bm{\xi} ds
            \right\}
    \end{split}
    \]
    where, because both $\bm{\eta}_\mathrm{tr}$ and $\bm{\eta}$ are determined by using the same strategy as in equation \eqref{eq:14}, \eqref{eq:15} or \eqref{eq:16}, the initial values satisfy
    \[
        \bm{\eta}_\mathrm{tr}=\bm{\eta}+\bm{\xi}
    \]
     and the integral term at the right hand side is manipulated as
    \[
        -\int_{t_1}^{t}\exp\left(\hat{\bm{A}}(t_1-s)\right)\hat{\bm{A}}\bm{\xi} ds = \exp\left(\hat{\bm{A}}(t_1-t)\right)\bm{\xi}-\bm{\xi}.
    \]

    Finally, a direct substitution and simple manipulation give the time response function
    \[
    \hat{\bm{y}}_\mathrm{tr}(t)
        = \exp\left(\hat{\bm{A}}(t-t_1)\right)
        \left\{ \bm{\eta} +\int_{t_1}^{t}
                \exp\left(\hat{\bm{A}}(t_1-s)\right)
                \left[ \hat{\bm{B}}\bm{u}(s)+\hat{\bm{c}} \right]ds
        \right\}
        +\bm{\xi}
        =\hat{\bm{y}}(t)+\bm{\xi}
    \]
    and the restored fitting and forecasting values
    \[
        \hat{\bm{x}}_\mathrm{tr}(t_k)
        =\frac{\hat{\bm{y}}_\mathrm{tr}(t_k)-\hat{\bm{y}}_\mathrm{tr}(t_{k-1})}{t_k-t_{k-1}}
        =\frac{\hat{\bm{y}}(t_k)-\hat{\bm{y}}(t_{k-1})}{t_k-t_{k-1}}
        =\bm{x}(t_k),~ k=2,3,\cdots,n+r.
    \]

    Lemma \ref{lem:01} shows that the translation transformation of the Cusum series has no influence on the forecasting results. Note that the translation transformation of the Cusum series is equivalent to adding $\bm{\xi}$ to the first element in the original time series ($\bm{X}_\mathrm{tr}(t)=\left\{ \bm{x}(t_1)+\bm{\xi}, \bm{x}(t_2), \cdots, \bm{x}(t_n) \right\}$). Thus, without changing performance, the Cusum operator in Definition \ref{def:01} can be defined also as
    $\bm{y}(t_1)=\bm{\xi}$ and $\bm{y}(t_k)=\bm{\xi}+\sum\nolimits_{i=2}^{k} h_k \bm{x}(t_k)$ when $k\geq 2$.
\end{proof}

\begin{theorem}\label{them:01}
    Let
    \begin{equation}\label{eq:int}
        \bm{y}(t)=\bm{\xi}+\int_{t_1}^{t} \bm{x}(s) ds, ~ t \geq t_1
    \end{equation}
    where $\bm{\xi}\in \mathbb{R}^d$ is a real vector. Then the ordinary differential equation \eqref{eq:gm} with initial value $\bm{y}(t_1)=\bm{\xi}$ is equivalent to the reduced-order ordinary differential equation
    \begin{equation}\label{eq:de}
        \frac{d}{dt} \bm{x}(t)=\bm{A}\bm{x}(t)+\bm{B} \frac{d}{dt} \bm{u}(t), ~ t \geq t_1
    \end{equation}
    with initial value $\bm{x}(t_1)=\bm{A}\bm{\xi}+\bm{B}\bm{u}(t_1)+\bm{c}$.
\end{theorem}

\begin{proof}
    The proof is divided into two parts. Firstly, we prove that equation \eqref{eq:gm} can be reduced to equation \eqref{eq:de}. By substituting equation \eqref{eq:int} into \eqref{eq:gm}, one has
    \begin{equation}\label{eq:19}
        \bm{x}(t)=\bm{A}\left(\bm{\xi}+\int_{t_1}^{t} \bm{x}(s) ds\right)+\bm{B}\bm{u}(t)+\bm{c}, ~ t \geq t_1.
    \end{equation}

    Differentiating both sides of equation \eqref{eq:19} with respect to $t$ gives the desired equation \eqref{eq:de}; then, combining equation \eqref{eq:int} and the closed-form solution in equation \eqref{eq:13}, gives the initial value expressed as
    \[
        \bm{x}(t_1)=\frac{d}{dt}\bm{y}(t)\Big|_{t=t_1}=\bm{A}\bm{\xi}+\bm{B}\bm{u}(t_1)+\bm{c}.
    \]

    Conversely, integrating both sides of equation \eqref{eq:de} gives
    \begin{equation}\label{eq:im}
        \bm{x}(t)=\bm{A}\int_{t_1}^{t}\bm{x}(s)ds+
        \bm{B}\bm{u}(t)-\bm{B}\bm{u}(t_1)+
        \bm{x}(t_1), ~ t \geq t_1
    \end{equation}
    where, by differentiating both sides of equation \eqref{eq:int} with respect to $t$, the left side is
    \[
        \bm{x}(t)=\frac{d}{dt}\left(\bm{\xi}+\int_{t_1}^{t} \bm{x}(s) ds\right)=\frac{d}{dt}\bm{y}(t), ~ t \geq t_1
    \]
    and by substituting $\bm{y}(t)-\bm{\xi}$ in equation \eqref{eq:int} for $\int_{t_1}^{t}\bm{x}(s)ds$ in equation \eqref{eq:im}, the right side is
    \[
        \bm{A}\int_{t_1}^{t}\bm{x}(s)ds+
        \bm{B}\bm{u}(t)-\bm{B}\bm{u}(t_1)+
        \bm{x}(t_1)
        = \bm{A}\bm{y}(t)+\bm{B}\bm{u}(t)+\bm{x}(t_1)-\bm{A}\bm{\xi}-\bm{B}\bm{u}(t_1).
    \]

    Substituting the initial value $\bm{x}(t_1)=\bm{c}+\bm{B}\bm{u}(t_1)+\bm{A}\bm{\xi}$ into equation \eqref{eq:im} then yields the desired differential equation \eqref{eq:gm}.
\end{proof}

Theorem \ref{them:01} shows that by introducing the integral operator \eqref{eq:int}, equation \eqref{eq:gm} can always be reduced to the simple form \eqref{eq:de} and, in the other direction, equation \eqref{eq:de} can express \eqref{eq:gm} by selecting an appropriate initial value $\bm{x}(t_1)=\bm{c}+\bm{B}\bm{u}(t_1)+\bm{A}\bm{\xi}$. In this sense, equation \eqref{eq:gm} has $d$ redundant degrees of freedom (the parameter vector $\bm{c}$) and it can be eliminated by using the reduced-order equation \eqref{eq:de} to model the original time series directly.
Here, we use an example to show the order reduction process. Supposing the ordinary differential equation \eqref{eq:gm} is
$
    \frac{d}{dt}y(t)=ay(t) + b_2 t^2 + b_1 t + c, ~ y(0)=\xi, ~ t\geq 0
$,
whose closed-form solution is
\[
    y(t) = \left(\xi + \frac{c}{a} + \frac{b_1}{a^2} + \frac{2 b_2}{a^3}\right) \exp(at) -
        \frac{b_2}{a} t^2 -
        \left(\frac{b_1}{a}+\frac{2 b_2}{a^2}\right) t -
        \left(\frac{c}{a} + \frac{b_1}{a^2} + \frac{2 b_2}{a^3}\right), ~ t\geq 0 .
\]
If $y(t)$ satisfies equation \eqref{eq:int}, i.e., $y(t)=\xi+\int_{0}^{t} x(s) ds$, then it is easy to derive that  the equivalent reduced-order equation \eqref{eq:de} is
$
    \frac{d}{dt}x(t)=ax(t) + 2b_2 t + b_1, ~ x(0)=a\xi+c, ~ t\geq 0
$,
whose closed-form solution is
\[
    x(t)=\left( a\xi+c+\frac{b_1}{a}+\frac{2b_2}{a^2} \right)\exp(at)-\frac{2b_2}{a}t-\left( \frac{b_1}{a}+\frac{2b_2}{a^2} \right), ~ t\geq 0.
\]

Since Lemma \ref{lem:01} shows that $\bm{\xi}$ has no influence on modelling performance then, without loss of generality, let $\bm{\xi}=\bm{x}(t_1)$ in order to be consistent with the intuitive understanding of the Cusum operator. In this way, the initial values satisfy
$
    \bm{y}(t_1)=\bm{\xi}=\bm{x}(t_1)
$
and the redundant parameter vector satisfy
$
    \bm{c}=\left(\bm{I}-\bm{A}\right)\bm{y}(t_1)-\bm{B}\bm{u}(t_1).
$
This provides an alternative strategy for initial value selection, i.e.,
$
     \hat{\bm{y}}(t_1) = (\bm{I}-\hat{\bm{A}} )^\mathrm{-1} (\hat{\bm{c}}+\hat{\bm{B}}\bm{u}(t_1) ).
$

\subsection{Discretization of integral operator for explaining the cumulative sum operator}\label{sec:4-1}

The integral operator \eqref{eq:int} plays an important role in the simplification from equation \eqref{eq:gm} to \eqref{eq:de}. Since the observations of the state vector $\bm{x}(t)$ at time points $\{t_k\}_{k=1}^{n}$ are known in advance, the finite integral in equation \eqref{eq:int} can be approximated using the piecewise interpolation methods, such as the piecewise constant integration formula
\begin{equation}\label{eq:22}
    \bm{y}_\mathrm{cnt}(t_k)=\bm{x}(t_1)+\int_{t_1}^{t_k} \bm{x}(s) ds
    =\begin{cases}
        \bm{x}(t_1), & k=1 \\
        \bm{x}(t_1)+\sum\limits_{i=2}^{k}{h_i}{\bm{x}(t_i)}, & k=2,3,\cdots,n \\
    \end{cases}
\end{equation}
and the piecewise linear integration formula
\begin{equation}\label{eq:23}
    \bm{y}_\mathrm{lnt}(t_k)=\bm{x}(t_1)+\int_{t_1}^{t_k} \bm{x}(s) ds
    =\begin{cases}
        \bm{x}(t_1), & k=1 \\
        \bm{x}(t_1)+\dfrac{1}{2}\sum\limits_{i=2}^{k}{h_i}{\bm{x}(t_{i-1})}+\dfrac{1}{2}\sum\limits_{i=2}^{k}{h_i}{\bm{x}(t_i)}, & k=2,3,\cdots,n \\
    \end{cases}.
\end{equation}

\begin{remark}\label{rem:04}
    The piecewise constant integral formula in equation \eqref{eq:22} is equivalent to the Cusum operator in Definition \ref{def:01}, indicating that the Cusum operator is a numerical discretization of the integral operator \eqref{eq:int} and, in turn, the latter is the continuous generalization of the former.
\end{remark}

Note that equation \eqref{eq:23} is a second-order formula and thus has higher accuracy than the first-order \eqref{eq:22}. Furthermore, other higher-order numerical integration approaches can also be used to further improve the accuracy when calculating the finite integration in equation \eqref{eq:int}, such as the Simpson rule \cite{ding2021forecasting}.

\subsection{Integral matching for simultaneous structural parameter and initial value estimation}

Integral matching-based model uses the reduced-order equation \eqref{eq:de} to fit the original time series directly, and the key is to estimate the structural parameter and initial value from sampled time-series data.

Theorem \ref{lem:01} shows that integrating both sides of equation \eqref{eq:de} gives the definite integral equation \eqref{eq:im}. By substituting $\bm{x}(t_1)=\bm{\eta}$ into equation \eqref{eq:im}, one has
\[
    \bm{x}(t)=\bm{A}\int_{t_1}^{t}\bm{x}(s)ds+
    \bm{B}\left(\bm{u}(t)-\bm{u}(t_1)\right)+
    \bm{\eta}, ~ t \geq t_1;
\]
and then, using the numerical integration formula \eqref{eq:23} to approximate the integrals, we obtain
\begin{equation}\label{eq:24}
    \bm{x}(t_k)=\bm{A}\left[ \bm{y}_\mathrm{lnt}(t_k)- \bm{x}(t_1)\right]+
    \bm{B}\left[\bm{u}(t_k)-\bm{u}(t_1)\right]+
    \bm{\eta}+\bm{\epsilon}(t_k)
\end{equation}
where $\bm{\epsilon}(t_k)$ is the sum of numerical error and noise error at each time instant.

Similar to the treatment process in section \ref{sec:3}, by substituting $k=2,3,\cdots,n$ and the observed data into equation \eqref{eq:24}, and then arranging the $n-1$ algebraic equations into a matrix form, one has
\begin{equation}\label{eq:25}
    \bm{X}=\bm{\Omega}(\bm{x}, \bm{u}) \bm{\Pi}+\bm{\Upsilon}
\end{equation}
where
\[
    \bm{\Pi}=
        \begin{bmatrix}
            \bm{A}^\mathsf{T} \\
            \bm{B}^\mathsf{T} \\
            \bm{\eta}^\mathsf{T}
        \end{bmatrix}, ~
    \bm{\Omega}(\bm{x}, \bm{u})=
        \begin{bmatrix}
            \bm{y}^\mathsf{T}_\mathrm{lnt}(t_2)-\bm{x}^\mathsf{T}(t_1) & \bm{u}^\mathsf{T}(t_2)-\bm{u}^\mathsf{T}(t_1) & 1 \\
            \bm{y}^\mathsf{T}_\mathrm{lnt}(t_3)-\bm{x}^\mathsf{T}(t_1) & \bm{u}^\mathsf{T}(t_3)-\bm{u}^\mathsf{T}(t_1) & 1 \\
            \vdots & \vdots & \vdots \\
            \bm{y}^\mathsf{T}_\mathrm{lnt}(t_n)-\bm{x}^\mathsf{T}(t_1) & \bm{u}^\mathsf{T}(t_n)-\bm{u}^\mathsf{T}(t_1) & 1
        \end{bmatrix},~
    \bm{\Upsilon}=
        \begin{bmatrix}
            \bm{\epsilon}^\mathsf{T}(t_2) \\
            \bm{\epsilon}^\mathsf{T}(t_3) \\
            \vdots \\
            \bm{\epsilon}^\mathsf{T}(t_n)
        \end{bmatrix}.
\]

And by minimizing the objective function
$
    \mathcal{L}(\bm{\Pi})=\left\|\bm{X}-\bm{\Omega}(\bm{x}, \bm{u}) \bm{\Pi}\right\|_\mathsf{F}^2
$,
one has the least-square parameter matrix
\begin{equation}\label{eq:26}
    \hat{\bm{\Pi}}
    =\left(\bm{\Omega}^\mathsf{T}(\bm{x},\bm{u})\bm{\Omega}(\bm{x},\bm{u})\right)^{-1}
    \bm{\Omega}^\mathsf{T}(\bm{x},\bm{u})\bm{X}.
\end{equation}

Substituting the estimates of the structural parameters ($\hat{\bm{A}}$ and $\hat{\bm{B}}$) and initial value ($\hat{\bm{\eta}}$) into the general solution to the reduced-order equation \eqref{eq:de}, gives the time response function
\begin{equation}\label{eq:27}
\begin{split}
    \hat{\bm{x}}(t)
        &=
        \exp\left(\hat{\bm{A}}(t-t_1)\right)
        \left\{
            \hat{\bm{\eta}}+
            \exp\left(\hat{\bm{A}}(t_1-t)\right)\hat{\bm{B}}\bm{u}(t)-
            \hat{\bm{B}}\bm{u}(t_1)+
            \int_{t_1}^{t}\exp\left(\hat{\bm{A}}(t_1-s)\right)\hat{\bm{A}}\hat{\bm{B}}\bm{u}(s)ds
        \right\}
\end{split}
\end{equation}
and then the fitting and forecasting values of the original time series can be obtained by substituting the time points $\left\{t_k\right\}_{k=1}^{n+r}$ into equation \eqref{eq:27}.

\section{The relationship between grey models and integral matching-based models}\label{sec:a-1}

In this section we analyse the relationship between grey forecasting models and integral matching-based differential equation models from both the parameter estimation and the modelling procedures perspectives.
For convenience, the original time series is assumed to be equally spaced, i.e., $h_k=h$, $k=2,3,\cdots,n$. The estimated parameters of grey models are denoted as $\hat{\bm{A}}_\mathrm{g}$, $\hat{\bm{B}}_\mathrm{g}$ and $\hat{\bm{c}}_\mathrm{g}$, and
those of integral matching-based ones are denoted as $\hat{\bm{A}}_\mathrm{m}$, $\hat{\bm{B}}_\mathrm{m}$  and $\hat{\bm{\eta}}_\mathrm{m}$.

\subsection{Quantitative relationship between the estimated parameters}\label{sec:a4-1}

\begin{proposition}\label{pro:1}
    Let the independent vector be $\bm{u}(t)=0$. The estimated parameters satisfy $\hat{\bm{A}}_\mathrm{g}=\hat{\bm{A}}_\mathrm{m}$ and $\hat{\bm{\eta}}_\mathrm{m}=\hat{\bm{c}}_\mathrm{g}+\hat{\bm{A}}_\mathrm{g} \bm{x}(t_1)-\frac{h}{2} \hat{\bm{A}}_\mathrm{g} \bm{x}(t_1)$.
\end{proposition}

\begin{proof}
    From equations \eqref{eq:12} and \eqref{eq:26}, it is easy to show that the estimated parameters are
    \[
        \begin{bmatrix}
            \hat{\bm{A}}_\mathrm{g}^\mathsf{T} \\
            \hat{\bm{c}}_\mathrm{g}^\mathsf{T}
        \end{bmatrix}
        =\left(\bm{Y}_\mathrm{g}^\mathsf{T}\bm{Y}_\mathrm{g}\right)^{-1} \bm{Y}_\mathrm{g}^\mathsf{T}\bm{X}
    \quad
    \text{and}
    \quad
        \begin{bmatrix}
            \hat{\bm{A}}_\mathrm{m}^\mathsf{T} \\
            \hat{\bm{\eta}}_\mathrm{m}^\mathsf{T}
        \end{bmatrix}
        =\left(\bm{Y}_\mathrm{m}^\mathsf{T}\bm{Y}_\mathrm{m}\right)^{-1} \bm{Y}_\mathrm{m}^\mathsf{T}\bm{X}
    \]
    where
    \[
    \bm{Y}_\mathrm{g}=
        \begin{bmatrix}
            \frac{\bm{y}^\mathsf{T}(t_1)+\bm{y}^\mathsf{T}(t_2)}{2} & 1 \\
            \frac{\bm{y}^\mathsf{T}(t_2)+\bm{y}^\mathsf{T}(t_3)}{2} & 1 \\
            \vdots & \vdots  \\
            \frac{\bm{y}^\mathsf{T}(t_{n-1})+\bm{y}^\mathsf{T}(t_n)}{2} & 1
        \end{bmatrix}
    \quad
    \text{and}
    \quad
    \bm{Y}_\mathrm{m}=
        \begin{bmatrix}
            \bm{y}^\mathsf{T}_\mathrm{lnt}(t_2)-\bm{x}^\mathsf{T}(t_1) & 1 \\
            \bm{y}^\mathsf{T}_\mathrm{lnt}(t_3)-\bm{x}^\mathsf{T}(t_1) & 1 \\
            \vdots & \vdots \\
            \bm{y}^\mathsf{T}_\mathrm{lnt}(t_n)-\bm{x}^\mathsf{T}(t_1) & 1
        \end{bmatrix}.
    \]

    Because the original time series is equally spaced, equation \eqref{eq:23} can be rewritten as
    \[
        \bm{y}_\mathrm{lnt}(t_k)=
            \frac{1}{2}\bm{y}(t_{k-1})+\frac{1}{2}\bm{y}(t_{k})+\frac{h}{2}\bm{x}(t_1),
            ~ k=2,3,\cdots,n
    \]
    and thus we have
    $
        \bm{Y}_\mathrm{m}
            =\bm{Y}_\mathrm{g}\bm{Q}
    $,
    where $\bm{I}_d$ is the $d$th-order unit matrix,
    \[
        \bm{Q}=\begin{bmatrix}
            \bm{I}_d & 0 \\
            \frac{h-2}{2} \bm{x}^\mathsf{T}(t_1) &  1
        \end{bmatrix}
        \quad
        \text{and}
        \quad
        \bm{Q}^{-1}=\begin{bmatrix}
            \bm{I}_d & 0 \\
            -\frac{h-2}{2} \bm{x}^\mathsf{T}(t_1) &  1
        \end{bmatrix}
    \]

    Finally, a simple substitution gives the desired equation
    \[
    \begin{bmatrix}
        \hat{\bm{A}}_\mathrm{m}^\mathsf{T} \\
        \hat{\bm{\eta}}_\mathrm{m}^\mathsf{T}
        \end{bmatrix}
        =\left(\left(\bm{Y}_\mathrm{g}\bm{Q}\right)^\mathsf{T} \left(\bm{Y}_\mathrm{g}\bm{Q}\right) \right)^{-1}
         \left(\bm{Y}_\mathrm{g}\bm{Q}\right)^\mathsf{T} \bm{X}
        =\bm{Q}^{-1}
            \begin{bmatrix}
                \hat{\bm{A}}_\mathrm{g}^\mathsf{T} \\
                \hat{\bm{c}}_\mathrm{g}^\mathsf{T}
            \end{bmatrix}
        =\begin{bmatrix}
            \hat{\bm{A}}_\mathrm{g}^\mathsf{T} \\
            \hat{\bm{c}}_\mathrm{g}^\mathsf{T}+
            \bm{x}^\mathsf{T}(t_1) \hat{\bm{A}}_\mathrm{g}^\mathsf{T}- \frac{h}{2} \bm{x}^\mathsf{T}(t_1) \hat{\bm{A}}_\mathrm{g}^\mathsf{T}
        \end{bmatrix}.
    \]

    Proposition \ref{pro:1} shows that the estimated matrices satisfy $\hat{\bm{A}}_\mathrm{g}=\hat{\bm{A}}_\mathrm{m}$ and $\hat{\bm{c}}_\mathrm{g}=\hat{\bm{\eta}}_\mathrm{m}-\hat{\bm{A}}_\mathrm{g} \bm{x}(t_1)+\frac{h}{2} \hat{\bm{A}}_\mathrm{g} \bm{x}(t_1)$, showing that grey models can be viewed as equivalent forms of integral matching-based models by applying a translation transformation (translation coefficient is $\bm{\xi}=\frac{h-2}{2} \bm{x}(t_1)$ in Lemma \ref{lem:01})  to the Cusum series, and considering the impact of translation transformation on the estimates is $\frac{2-h}{2} \hat{\bm{A}}_\mathrm{g} \bm{x}(t_1)$ in $\hat{\bm{c}}_\mathrm{g}$.
\end{proof}

If the independent vector $\bm{u}(t)\neq 0$, the quantitative relationship is not that obvious. By partitioning the estimated parameter matrices in equations \eqref{eq:12} and \eqref{eq:26} as
$
\hat{\bm{\Xi}}_\mathrm{g}^\mathsf{T} =
    \begin{bmatrix}
        \hat{\bm{A}}_\mathrm{g} ~ \hat{\bm{c}}_\mathrm{g}  ~ | ~ \hat{\bm{B}}_\mathrm{g}
    \end{bmatrix}
$
and
$
\hat{\bm{\Pi}}_\mathrm{m}^\mathsf{T} =
    \begin{bmatrix}
        \hat{\bm{A}}_\mathrm{m} ~ \hat{\bm{\eta}}_\mathrm{m}  ~ | ~ \hat{\bm{B}}_\mathrm{m}
    \end{bmatrix},
$
then we have
\begin{equation}\label{eq:31}
    \hat{\bm{\Xi}}_\mathrm{g}
    =\left(
        \begin{bmatrix}
            {\bm{Y}}_\mathrm{g}^\mathsf{T} \\
            {\bm{U}}_\mathrm{g}^\mathsf{T}
        \end{bmatrix}
        \begin{bmatrix}
            {\bm{Y}}_\mathrm{g} & {\bm{U}}_\mathrm{g}
        \end{bmatrix}
    \right)^{-1}
    \begin{bmatrix}
        {\bm{Y}}_\mathrm{g}^\mathsf{T} \\
        {\bm{U}}_\mathrm{g}^\mathsf{T}
    \end{bmatrix} \bm{X}
    =
    \begin{bmatrix}
        {\bm{Y}}_\mathrm{g}^\mathsf{T} {\bm{Y}}_\mathrm{g} & {\bm{Y}}_\mathrm{g}^\mathsf{T} {\bm{U}}_\mathrm{g}  \\
        {\bm{U}}_\mathrm{g}^\mathsf{T} {\bm{Y}}_\mathrm{g} & {\bm{U}}_\mathrm{g}^\mathsf{T} {\bm{U}}_\mathrm{g}
    \end{bmatrix} ^{-1}
    \begin{bmatrix}
        {\bm{Y}}_\mathrm{g}^\mathsf{T}\bm{X} \\
        {\bm{U}}_\mathrm{g}^\mathsf{T}\bm{X}
    \end{bmatrix}
\end{equation}
and
\begin{equation}\label{eq:32}
    \hat{\bm{\Pi}}_\mathrm{m}=
    \left(
        \begin{bmatrix}
            {\bm{Y}}_\mathrm{m}^\mathsf{T} \\
            {\bm{U}}_\mathrm{m}^\mathsf{T}
        \end{bmatrix}
        \begin{bmatrix}
            {\bm{Y}}_\mathrm{m} & {\bm{U}}_\mathrm{m}
        \end{bmatrix}
    \right)^{-1}
    \begin{bmatrix}
        {\bm{Y}}_\mathrm{m}^\mathsf{T} \\
        {\bm{U}}_\mathrm{m}^\mathsf{T}
    \end{bmatrix} \bm{X}
    =
    \begin{bmatrix}
        {\bm{Y}}_\mathrm{m}^\mathsf{T} {\bm{Y}}_\mathrm{m} & {\bm{Y}}_\mathrm{m}^\mathsf{T} {\bm{U}}_\mathrm{m}  \\
        {\bm{U}}_\mathrm{m}^\mathsf{T} {\bm{Y}}_\mathrm{m} & {\bm{U}}_\mathrm{m}^\mathsf{T} {\bm{U}}_\mathrm{m}
    \end{bmatrix} ^{-1}
    \begin{bmatrix}
        {\bm{Y}}_\mathrm{m}^\mathsf{T}\bm{X} \\
        {\bm{U}}_\mathrm{m}^\mathsf{T}\bm{X}
    \end{bmatrix}
\end{equation}
where
\[
    \bm{U}_\mathrm{g}=\frac{1}{2}
    \begin{bmatrix}
        \bm{u}^\mathsf{T}(t_2) + \bm{u}^\mathsf{T}(t_1) \\
        \bm{u}^\mathsf{T}(t_3) + \bm{u}^\mathsf{T}(t_2) \\
         \vdots \\
        \bm{u}^\mathsf{T}(t_n) + \bm{u}^\mathsf{T}(t_{n-1})
    \end{bmatrix}
    =\begin{bmatrix}
            \bm{u}^\mathsf{T}(t_2) \\
            \bm{u}^\mathsf{T}(t_3) \\
             \vdots \\
            \bm{u}^\mathsf{T}(t_n)
        \end{bmatrix}
    +\frac{1}{2}
        \begin{bmatrix}
            \bm{u}^\mathsf{T}(t_1) - \bm{u}^\mathsf{T}(t_2) \\
            \bm{u}^\mathsf{T}(t_2) - \bm{u}^\mathsf{T}(t_3) \\
            \vdots \\
            \bm{u}^\mathsf{T}(t_{n-1}) - \bm{u}^\mathsf{T}(t_n)
        \end{bmatrix}
\]
    \text{and}
\[
    \bm{U}_\mathrm{m}=
    \begin{bmatrix}
        \bm{u}^\mathsf{T}(t_2)-\bm{u}^\mathsf{T}(t_1) \\
        \bm{u}^\mathsf{T}(t_3)-\bm{u}^\mathsf{T}(t_1) \\
         \vdots \\
        \bm{u}^\mathsf{T}(t_n)-\bm{u}^\mathsf{T}(t_1)
    \end{bmatrix}
    =
    \begin{bmatrix}
        \bm{u}^\mathsf{T}(t_2) \\
        \bm{u}^\mathsf{T}(t_3) \\
        \vdots \\
        \bm{u}^\mathsf{T}(t_n)
    \end{bmatrix}
    -
    \begin{bmatrix}
        \bm{u}^\mathsf{T}(t_1) \\
        \bm{u}^\mathsf{T}(t_1) \\
        \vdots \\
        \bm{u}^\mathsf{T}(t_1)
    \end{bmatrix}
    =
    \bm{U}_\mathrm{g}
    +
    \frac{1}{2}
        \begin{bmatrix}
            \bm{u}^\mathsf{T}(t_2) - \bm{u}^\mathsf{T}(t_1)  \\
            \bm{u}^\mathsf{T}(t_3) - \bm{u}^\mathsf{T}(t_2)  \\
            \vdots \\
            \bm{u}^\mathsf{T}(t_n) - \bm{u}^\mathsf{T}(t_{n-1})
        \end{bmatrix}
    -
    \begin{bmatrix}
        \bm{u}^\mathsf{T}(t_1) \\
        \bm{u}^\mathsf{T}(t_1) \\
        \vdots \\
        \bm{u}^\mathsf{T}(t_1)
    \end{bmatrix}.
\]

By denoting $\bm{U}_\mathrm{m}=\bm{U}_\mathrm{g}+\bm{U}$ and combining $\bm{Y}_\mathrm{m}=\bm{Y}_\mathrm{g}\bm{Q}$ in Proposition \ref{pro:1}, then the matrices in equation \eqref{eq:32} can be rewritten as
\begin{equation}\label{eq:33}
    \begin{bmatrix}
        {\bm{Y}}_\mathrm{m}^\mathsf{T}\bm{X} \\
        {\bm{U}}_\mathrm{m}^\mathsf{T}\bm{X}
    \end{bmatrix}
    =
    \begin{bmatrix}
        {\bm{Q}}^\mathsf{T}  & {\bm{0}} \\
        {\bm{0}} & {\bm{I}}_{p}
    \end{bmatrix}
    \left(
        \begin{bmatrix}
            {\bm{Y}}_\mathrm{g}^\mathsf{T}\bm{X} \\
            {\bm{U}}_\mathrm{g}^\mathsf{T}\bm{X}
        \end{bmatrix}
        +
        \begin{bmatrix}
            \bm{0} \\
            \bm{U}^\mathsf{T} \bm{X}
        \end{bmatrix}
    \right)
\end{equation}
and
\begin{equation}\label{eq:34}
    \begin{bmatrix}
        {\bm{Y}}_\mathrm{m}^\mathsf{T} {\bm{Y}}_\mathrm{m} & {\bm{Y}}_\mathrm{m}^\mathsf{T} {\bm{U}}_\mathrm{m}  \\
        {\bm{U}}_\mathrm{m}^\mathsf{T} {\bm{Y}}_\mathrm{m} & {\bm{U}}_\mathrm{m}^\mathsf{T} {\bm{U}}_\mathrm{m}
    \end{bmatrix}
    =
    \begin{bmatrix}
        {\bm{Q}}^\mathsf{T}  & {\bm{0}} \\
        {\bm{0}} & {\bm{I}}_{p}
    \end{bmatrix}
    \left(
        \begin{bmatrix}
            {\bm{Y}}_\mathrm{g}^\mathsf{T} {\bm{Y}}_\mathrm{g} & {\bm{Y}}_\mathrm{g}^\mathsf{T} {\bm{U}}_\mathrm{g}  \\
            {\bm{U}}_\mathrm{g}^\mathsf{T} {\bm{Y}}_\mathrm{g} & {\bm{U}}_\mathrm{g}^\mathsf{T} {\bm{U}}_\mathrm{g}
        \end{bmatrix}
        +
        \begin{bmatrix}
            \bm{0} & {\bm{Y}}_\mathrm{g}^\mathsf{T} {\bm{U}}  \\
            \bm{U}^\mathsf{T} \bm{Y}_\mathrm{g} & \bm{W}
        \end{bmatrix}
    \right)
    \begin{bmatrix}
        {\bm{Q}}  & {\bm{0}} \\
        {\bm{0}} & {\bm{I}}_{p}
    \end{bmatrix}
\end{equation}
where $\bm{W}=\bm{U}_\mathrm{g}^\mathsf{T}\bm{U} + \bm{U}^\mathsf{T}\bm{U}_\mathrm{g} +\bm{U}^\mathsf{T}\bm{U}$.

By substituting equations \eqref{eq:33} and \eqref{eq:34} into equation \eqref{eq:32}, the estimates are obtained as
\begin{equation}\label{eq:35}
    \hat{\bm{\Pi}}_\mathrm{m}
    =
    \begin{bmatrix}
        {\bm{Q}}^{-1}  & {\bm{0}} \\
        {\bm{0}} & {\bm{I}}_{p}
    \end{bmatrix}
    \left(
        \begin{bmatrix}
            {\bm{Y}}_\mathrm{g}^\mathsf{T} {\bm{Y}}_\mathrm{g} & {\bm{Y}}_\mathrm{g}^\mathsf{T} {\bm{U}}_\mathrm{g}  \\
            {\bm{U}}_\mathrm{g}^\mathsf{T} {\bm{Y}}_\mathrm{g} & {\bm{U}}_\mathrm{g}^\mathsf{T} {\bm{U}}_\mathrm{g}
        \end{bmatrix}
        +
        \begin{bmatrix}
            \bm{0} & {\bm{Y}}_\mathrm{g}^\mathsf{T} {\bm{U}}  \\
            \bm{U}^\mathsf{T} \bm{Y}_\mathrm{g} & \bm{W}
        \end{bmatrix}
    \right)^{-1}
    \left(
        \begin{bmatrix}
            {\bm{Y}}_\mathrm{g}^\mathsf{T}\bm{X} \\
            {\bm{U}}_\mathrm{g}^\mathsf{T}\bm{X}
        \end{bmatrix}
        +
        \begin{bmatrix}
            \bm{0} \\
            \bm{U}^\mathsf{T} \bm{X}
        \end{bmatrix}
    \right)
\end{equation}
where, by using the Sherman-Morrison-Woodbury formula \cite{van1983matrix}, the inverse matrix can be rewritten as
\[
    \begin{bmatrix}
        {\bm{Y}}_\mathrm{g}^\mathsf{T} {\bm{Y}}_\mathrm{g} & {\bm{Y}}_\mathrm{g}^\mathsf{T} {\bm{U}}_\mathrm{g}  \\
        {\bm{U}}_\mathrm{g}^\mathsf{T} {\bm{Y}}_\mathrm{g} & {\bm{U}}_\mathrm{g}^\mathsf{T} {\bm{U}}_\mathrm{g}
    \end{bmatrix}^{-1}
    -
    \begin{bmatrix}
        {\bm{Y}}_\mathrm{g}^\mathsf{T} {\bm{Y}}_\mathrm{g} & {\bm{Y}}_\mathrm{g}^\mathsf{T} {\bm{U}}_\mathrm{g}  \\
        {\bm{U}}_\mathrm{g}^\mathsf{T} {\bm{Y}}_\mathrm{g} & {\bm{U}}_\mathrm{g}^\mathsf{T} {\bm{U}}_\mathrm{g}
    \end{bmatrix}^{-1}
    \left(
        \begin{bmatrix}
            {\bm{Y}}_\mathrm{g}^\mathsf{T} {\bm{Y}}_\mathrm{g} & {\bm{Y}}_\mathrm{g}^\mathsf{T} {\bm{U}}_\mathrm{g}  \\
            {\bm{U}}_\mathrm{g}^\mathsf{T} {\bm{Y}}_\mathrm{g} & {\bm{U}}_\mathrm{g}^\mathsf{T} {\bm{U}}_\mathrm{g}
        \end{bmatrix}^{-1}
        +
        \begin{bmatrix}
            \bm{0} & {\bm{Y}}_\mathrm{g}^\mathsf{T} {\bm{U}}  \\
            \bm{U}^\mathsf{T} \bm{Y}_\mathrm{g} & \bm{W}
        \end{bmatrix}^{-1}
    \right)
    \begin{bmatrix}
        {\bm{Y}}_\mathrm{g}^\mathsf{T} {\bm{Y}}_\mathrm{g} & {\bm{Y}}_\mathrm{g}^\mathsf{T} {\bm{U}}_\mathrm{g}  \\
        {\bm{U}}_\mathrm{g}^\mathsf{T} {\bm{Y}}_\mathrm{g} & {\bm{U}}_\mathrm{g}^\mathsf{T} {\bm{U}}_\mathrm{g}
    \end{bmatrix}^{-1}.
\]
Therefore, it follows that
{
\[
\begin{split}
    \hat{\bm{\Pi}}_\mathrm{m}
    & =
    \begin{bmatrix}
        {\bm{Q}}^{-1}  & {\bm{0}} \\
        {\bm{0}} & {\bm{I}}_{p}
    \end{bmatrix}
    \hat{\bm{\Xi}}_\mathrm{g} \\
    & +
    \begin{bmatrix}
        {\bm{Q}}^{-1}  & {\bm{0}} \\
        {\bm{0}} & {\bm{I}}_{p}
    \end{bmatrix}
    \begin{bmatrix}
        {\bm{Y}}_\mathrm{g}^\mathsf{T} {\bm{Y}}_\mathrm{g} & {\bm{Y}}_\mathrm{g}^\mathsf{T} {\bm{U}}_\mathrm{g}  \\
        {\bm{U}}_\mathrm{g}^\mathsf{T} {\bm{Y}}_\mathrm{g} & {\bm{U}}_\mathrm{g}^\mathsf{T} {\bm{U}}_\mathrm{g}
    \end{bmatrix}^{-1}
    \begin{bmatrix}
        \bm{0} \\
        \bm{U}^\mathsf{T} \bm{X}
    \end{bmatrix} \\
    & -
    \begin{bmatrix}
        {\bm{Q}}^{-1}  & {\bm{0}} \\
        {\bm{0}} & {\bm{I}}_{p}
    \end{bmatrix}
    \begin{bmatrix}
        {\bm{Y}}_\mathrm{g}^\mathsf{T} {\bm{Y}}_\mathrm{g} & {\bm{Y}}_\mathrm{g}^\mathsf{T} {\bm{U}}_\mathrm{g}  \\
        {\bm{U}}_\mathrm{g}^\mathsf{T} {\bm{Y}}_\mathrm{g} & {\bm{U}}_\mathrm{g}^\mathsf{T} {\bm{U}}_\mathrm{g}
    \end{bmatrix}^{-1}
    \left(
        \begin{bmatrix}
            {\bm{Y}}_\mathrm{g}^\mathsf{T} {\bm{Y}}_\mathrm{g} & {\bm{Y}}_\mathrm{g}^\mathsf{T} {\bm{U}}_\mathrm{g}  \\
            {\bm{U}}_\mathrm{g}^\mathsf{T} {\bm{Y}}_\mathrm{g} & {\bm{U}}_\mathrm{g}^\mathsf{T} {\bm{U}}_\mathrm{g}
        \end{bmatrix}^{-1}
        +
        \begin{bmatrix}
            \bm{0} & {\bm{Y}}_\mathrm{g}^\mathsf{T} {\bm{U}}  \\
            \bm{U}^\mathsf{T} \bm{Y}_\mathrm{g} & \bm{W}
        \end{bmatrix}^{-1}
    \right)
    \hat{\bm{\Xi}}_\mathrm{g}  \\
    &-
    \begin{bmatrix}
        {\bm{Q}}^{-1}  & {\bm{0}} \\
        {\bm{0}} & {\bm{I}}_{p}
    \end{bmatrix}
    \begin{bmatrix}
        {\bm{Y}}_\mathrm{g}^\mathsf{T} {\bm{Y}}_\mathrm{g} & {\bm{Y}}_\mathrm{g}^\mathsf{T} {\bm{U}}_\mathrm{g}  \\
        {\bm{U}}_\mathrm{g}^\mathsf{T} {\bm{Y}}_\mathrm{g} & {\bm{U}}_\mathrm{g}^\mathsf{T} {\bm{U}}_\mathrm{g}
    \end{bmatrix}^{-1}
    \left(
        \begin{bmatrix}
            {\bm{Y}}_\mathrm{g}^\mathsf{T} {\bm{Y}}_\mathrm{g} & {\bm{Y}}_\mathrm{g}^\mathsf{T} {\bm{U}}_\mathrm{g}  \\
            {\bm{U}}_\mathrm{g}^\mathsf{T} {\bm{Y}}_\mathrm{g} & {\bm{U}}_\mathrm{g}^\mathsf{T} {\bm{U}}_\mathrm{g}
        \end{bmatrix}^{-1}
        +
        \begin{bmatrix}
            \bm{0} & {\bm{Y}}_\mathrm{g}^\mathsf{T} {\bm{U}}  \\
            \bm{U}^\mathsf{T} \bm{Y}_\mathrm{g} & \bm{W}
        \end{bmatrix}^{-1}
    \right)
    \begin{bmatrix}
        {\bm{Y}}_\mathrm{g}^\mathsf{T} {\bm{Y}}_\mathrm{g} & {\bm{Y}}_\mathrm{g}^\mathsf{T} {\bm{U}}_\mathrm{g}  \\
        {\bm{U}}_\mathrm{g}^\mathsf{T} {\bm{Y}}_\mathrm{g} & {\bm{U}}_\mathrm{g}^\mathsf{T} {\bm{U}}_\mathrm{g}
    \end{bmatrix}^{-1}
    \begin{bmatrix}
        \bm{0} \\
        \bm{U}^\mathsf{T} \bm{X}
    \end{bmatrix}
\end{split}
\]
}
where the first term on the right hand side shows that
\[
    \hat{\bm{A}}_\mathrm{g}=\hat{\bm{A}}_\mathrm{m},~
    \hat{\bm{B}}_\mathrm{g}=\hat{\bm{B}}_\mathrm{m},~
    \hat{\bm{\eta}}_\mathrm{m}=\hat{\bm{c}}_\mathrm{g}-\frac{h-2}{2} \hat{\bm{A}}_\mathrm{g} \bm{x}(t_1)
\]
and the other three terms can be viewed as the calibration of the estimated parameter matrices by using the information in the matrix $\bm{U}$. In particular, the solution in equation \eqref{eq:35} yields that in Proposition \ref{pro:1} when the matrix is  $\bm{U}=\bm{0}$.

\subsection{Modelling procedures reconstruction of grey forecasting models}\label{sec:4-2}

The main objective of both grey models and integral matching-based models is the explanation of time-series data using ordinary differential equations. The modelling procedures of both models are summarized in Figure \ref{fig:0}.

\begin{figure}[!ht]
    \centering
    \includegraphics[scale=0.66, trim = 0 0 0 0, clip = true]{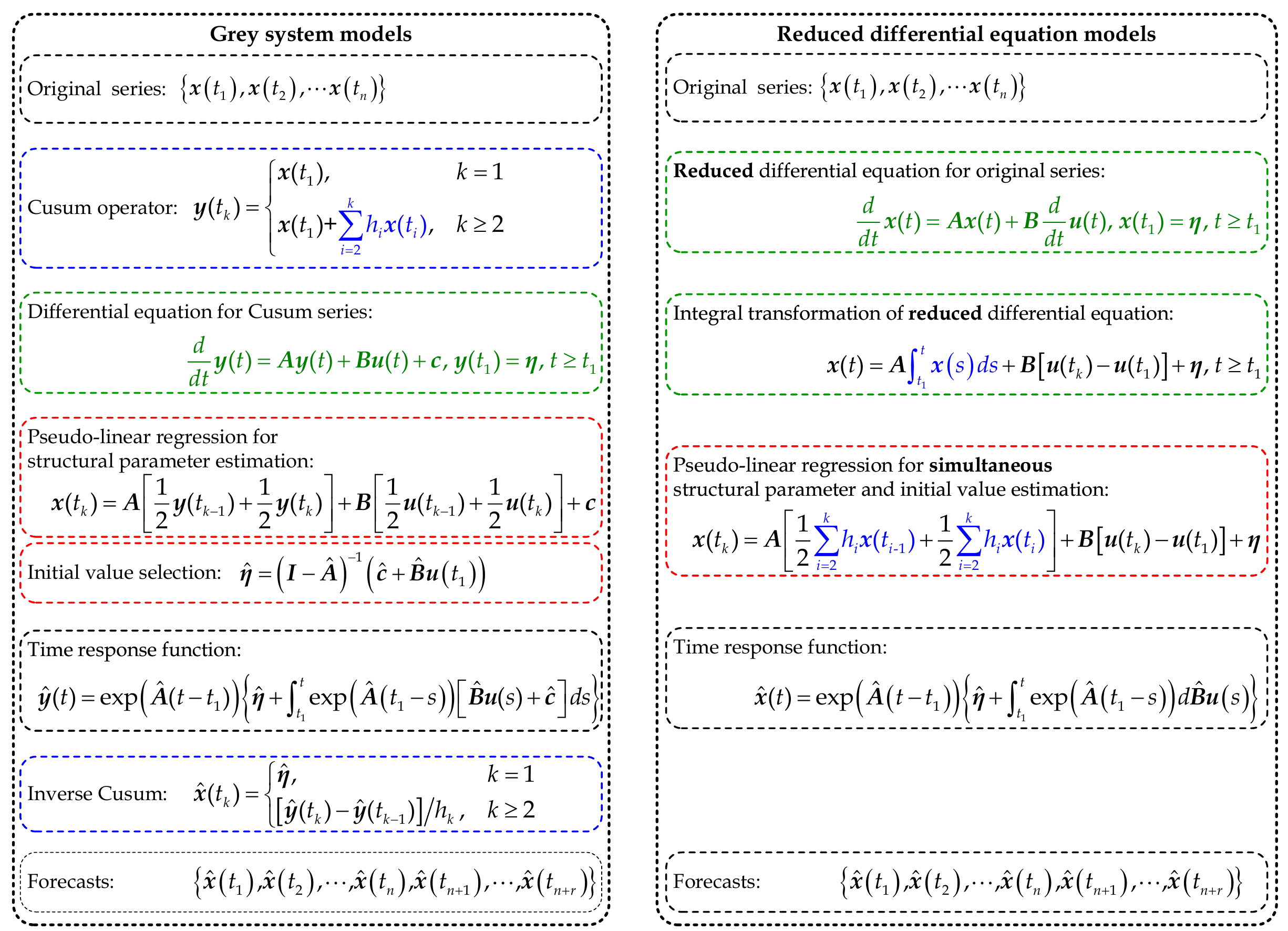} \\
    \caption{Reconstructed modelling procedures of grey models by integral matching-based models.}
    \label{fig:0}
\end{figure}

Figure \ref{fig:0} shows that the modelling procedures of integral matching-based models can be viewed as the simplification and reconstruction of grey models. In particular:
\begin{enumerate}
    \item [(i)]
        Grey models use ordinary differential equations to fit the Cusum series, while integral matching-based models employ reduced-order differential equations to explain and approximate the original series.
    \item [(ii)]
        Both of them use the trapezoid rule to discretize the ordinary differential equations into the discrete-time pseudo-regression expressions when estimating the parameters.
    \item [(iii)]
        The integral matching-based models obtain the estimates of structural parameters and initial values simultaneously, whereas the conventional grey models need to select a strategy to determine the initial value, such as the least-square strategy in equation \eqref{eq:16}, which may well introduce additional errors \cite{wei2019data}.
    \item [(iv)]
        Both of them use the time response function to compute the fitting and forecasting values. The grey models need to use the inverse Cusum operator to restore the forecasting results of the Cusum series, which makes this kind of model more tedious than the integral matching-based ones.
\end{enumerate}

The above modelling procedures show that both grey and integral matching-based models are constructed under the assumption that each independent variable ${u}_\iota(t)$ is known in advance, such as the deterministic function of time in Remark \ref{rem:1}. In some cases, there exists time-lag effect between the state variable and independent variable, such as the cases in Remark \ref{rem:2} having $u_\iota(t)=y_\iota(t-\tau_\iota h)$ where $\tau_\iota\geq 1$ is the time delay order and $h$ is the time interval. It is worth noting, however, that the possibility of complex `forcing variables' occurs often in science and economics. In forecasting terms, such complex variables need to be accounted for as `leading indicators' over the forecasting horizon. In this situation, these input variables will themselves have to be modelled and forecasted over the forecasting horizon. For example, univariate or multivariate time series methods can be applied, such as the dynamic harmonic regression \cite{young1999dynamic,tych2002unobserved}; see \cite{Young:2017qf} for a comprehensive example.

\section{Simulations}\label{sec:5}

In this section large-scale Monte Carlo simulations are conducted to compare grey models and integral matching-based models from both parameter estimation and forecasting accuracy perspectives.

The time-series data are generated from the observation equation
\begin{equation}\label{eq:a1}
    \tilde{\bm{x}}(t_k)=\begin{cases}
        \bm{x}(t_k)+\bm{e}(k), & k=1,2,\cdots,n \quad \leftarrow \text{in-sample series} \\
        \bm{x}(t_k), & k=n+1,\cdots, n+r        \quad \leftarrow \text{out-of-sample series} \\
    \end{cases}
\end{equation}
where $\bm{x}(t)$ is the solution to the state equation \eqref{eq:de}; $\bm{e}(k) \sim \mathcal{N}(\bm{0},\bm{\Sigma})$ is a $d$-dimensional normal distribution with mean $\bm{0}$ and covariance matrix $\bm{\Sigma}=\sigma^2\bm{I}_d$.
The fitting and forecasting errors are measured by the in-sample and out-of-sample mean absolute percentage error criteria, respectively expressed as
\[
    \mathrm{MAPE_{in}}[x_\ell]
    =\frac{1}{n}\sum_{k=1}^{n} \left| \frac{\hat{x}_\ell(t_k)-\tilde{x}_\ell(t_k)}{\tilde{x}_\ell(t_k)} \right|
        \times 100\%
\text{~ and ~}
    \mathrm{MAPE_{out}}[x_\ell]
    =\frac{1}{r}\sum_{k=n+1}^{n+r} \left| \frac{\hat{x}_\ell(t_k)-{x}_\ell(t_k)}{{x}_\ell(t_k)} \right|
        \times 100\%
\]
where $x_\ell(t)$, $\ell= 1,2,\cdots,d$, is the $\ell$-th component of the state variable $\bm{x}(t)$.

\subsection{Data generating process}

To analyse the effect of sample size and noise level on modelling results, the sample size $n$ and standard deviation $\sigma$ are set to different values. Here, we select a 2-dimensional autonomous system as a comprehensive example. Simulations of single-variable and other multi-variable models can be conducted by following similar steps.
Let the state equation \eqref{eq:de} be
\begin{align*}
    \begin{cases}
        \frac{d}{dt}x_1 (t) = -0.25 x_{1}(t)+0.70 x_{2}(t) \\
        \frac{d}{dt}x_2 (t) = 0.75 x_{1}(t)-0.25 x_{2}(t)
    \end{cases}, ~
    \begin{bmatrix}
        x_1(0) \\
        x_2(0)
    \end{bmatrix}
    =\begin{bmatrix}
        1.20 \\
        0.35
    \end{bmatrix},~
    t\in[0, 5]
\end{align*}
which indicates that the true values of structural parameter matrix and initial vector are
\[
\bm{A}
=\begin{bmatrix}
    a_{1,1} & a_{1,2} \\
    a_{2,1} & a_{2,2}
\end{bmatrix}
=\begin{bmatrix}
    -0.25 & 0.70 \\
    0.75 & -0.25
\end{bmatrix},~
\bm{\eta}
=\begin{bmatrix}
    \eta_1 \\
    \eta_2
\end{bmatrix}
=\begin{bmatrix}
    1.20 \\
    0.35
\end{bmatrix}.
\]

Taking the time interval as $h\in [0.25, 0.10, 0.05]$, then the observations are sampled at every time interval of $h$ in the range of $t\in [0, 5]$, thereby generating $n=\left[\frac{5}{h}\right]+1\in \{21, 51, 101\}$ samples; using the signal-to-noise ratio $snr \in \{2.5, 3.5, 5.0\}$ to measure the noise level, then the standard deviation is $\sigma=\frac{1}{\sqrt{snr}}\sqrt{\mathrm{Var}(x_\ell(t))}$. By substituting $n$ and $\sigma$ into equation \eqref{eq:a1}, a total of $3\times 3=9$ scenarios are obtained and in each sample size and signal-to-noise ratio combination scenario, 1000 realizations are replicated.

Before reporting the modelling results, it should be noted from Theorem \ref{them:01} that the equivalent equation  utilized in grey model \eqref{eq:im} is
\begin{align*}
    \begin{cases}
        \frac{d}{dt}y_1 (t) = -0.25 y_{1}(t)+0.70 y_{2}(t)+1.2550 \\
        \frac{d}{dt}y_2 (t) = 0.75 y_{1}(t)-0.25 y_{2}(t)-0.4625
    \end{cases},~
    \begin{bmatrix}
        y_1(0) \\
        y_2(0)
    \end{bmatrix}
    =\begin{bmatrix}
        1.20 \\
        0.35
    \end{bmatrix},~
    t\in[0, 5].
\end{align*}

\subsection{Evaluation of parameter estimation performance}\label{sec:5-2}

The estimations of structural parameters and initial vectors obtained from grey modelling and integral matching techniques are summarized in Table \ref{tbl:4}. Here, the structural parameter estimation matrices obtained from these two approaches are identical, validating Proposition \ref{pro:1}.

On the basis of the results in Table \ref{tbl:4}, we may roughly conclude the following: the estimates of parameters, including structural parameters and initial values, tend to their corresponding true values, and meanwhile, the sample standard deviations of estimates get smaller with the increase of sample size, partially implying the asymptotic behaviour of these two approaches; the estimates of the initial vector obtained from integral matching are superior to those obtained from grey modelling in sample mean and sample standard deviation terms, indicating that integral matching has higher accuracy and robustness.
To be specific, assuming that the signal-to-noise ratio remains fixed, then the 9 scenarios are reshaped to be 3 groups; and in each group, as the sample size increases, the sample standard deviations of the estimates reduce successively, validating the efficiency of the estimators; whereas the sample means of the estimates do not get closer to their true values monotonically with the increase of sample size, indicating the likely existence of asymptotic bias on the estimates.

\begin{table}[ht]
    \centering
    \footnotesize
    \begin{threeparttable}
        \caption{Sample means (sample standard deviations) of the estimated parameters obtained from grey modelling and integral matching approaches in 1000 runs with the sample size $n=21, 51, 101$ and  signal-to-noise ratio $snr=2.5, 3.5, 5.0$. The true parameters are $a_{1,1}=-0.25$, $a_{1,2}=0.70$, $a_{2,1}=0.75$, $a_{2,2}=-0.25$, $\eta_1=1.20$ and $\eta_2=0.35$.}
        \label{tbl:4}
        \setlength{\tabcolsep}{1.8mm}
        \begin{tabular}{l lll lll lll}
            \toprule
            $n$ & $snr$ &
            \multicolumn{4}{l}{Structural parameter matrix \textsuperscript{a}} &
            \multicolumn{2}{l}{Initial vector \textsuperscript{b}} &
            \multicolumn{2}{l}{Initial vector \textsuperscript{c}}  \\
            \cmidrule(lr){3-6} \cmidrule(lr){7-8} \cmidrule(lr){9-10}
        &   & $\hat{a}_{1,1}$ & $\hat{a}_{1,2}$ & $\hat{a}_{2,1}$ & $\hat{a}_{2,2}$
            & $\hat{\eta}_1$ & $\hat{\eta}_2$ & $\hat{\eta}_1$ & $\hat{\eta}_2$ \\ \hline
        21  & 2.5 & $-$0.249(0.606) & 0.699(0.601) & 0.745(0.655) & $-$0.245(0.651) & 1.301(0.681) & 0.509(0.711) & 1.195(0.330) & 0.346(0.368) \\
            & 3.5 & $-$0.251(0.309) & 0.701(0.306) & 0.743(0.334) & $-$0.244(0.332) & 1.278(0.343) & 0.486(0.358) & 1.201(0.168) & 0.351(0.188) \\
            & 5.0 & $-$0.250(0.151) & 0.700(0.150) & 0.744(0.164) & $-$0.245(0.162) & 1.268(0.168) & 0.476(0.175) & 1.201(0.082) & 0.351(0.092) \\
        51  & 2.5 & $-$0.255(0.342) & 0.706(0.341) & 0.742(0.371) & $-$0.241(0.370) & 1.252(0.452) & 0.426(0.465) & 1.200(0.175) & 0.350(0.197) \\
            & 3.5 & $-$0.254(0.174) & 0.704(0.174) & 0.744(0.189) & $-$0.244(0.188) & 1.238(0.229) & 0.411(0.236) & 1.201(0.089) & 0.351(0.100) \\
            & 5.0 & $-$0.252(0.085) & 0.702(0.085) & 0.747(0.092) & $-$0.247(0.092) & 1.231(0.112) & 0.404(0.115) & 1.201(0.044) & 0.351(0.049) \\
        101 & 2.5 & $-$0.236(0.224) & 0.686(0.223) & 0.764(0.242) & $-$0.264(0.241) & 1.223(0.371) & 0.384(0.377) & 1.194(0.115) & 0.343(0.129) \\
            & 3.5 & $-$0.244(0.115) & 0.694(0.114) & 0.756(0.124) & $-$0.257(0.123) & 1.218(0.189) & 0.379(0.192) & 1.197(0.059) & 0.347(0.066) \\
            & 5.0 & $-$0.247(0.056) & 0.697(0.056) & 0.753(0.061) & $-$0.253(0.060) & 1.215(0.093) & 0.376(0.094) & 1.199(0.029) & 0.349(0.032) \\
        \bottomrule
    \end{tabular}
    \begin{tablenotes}
        \item [a] The estimated structural parameter matrices obtained from grey modelling and integral matching are identical.
        \item [b] The estimated initial vector of grey modelling is obtained from
            $\hat{\bm{\eta}}=(\bm{I}_2-\hat{\bm{A}} )^\mathrm{-1}\hat{\bm{c}}$.
        \item [c] The estimated initial vector corresponds to integral matching.
    \end{tablenotes}
    \end{threeparttable}
\end{table}

\begin{figure}[!ht]
    \centering
    \includegraphics[scale=0.9, trim = 0 0 0 0, clip = true]{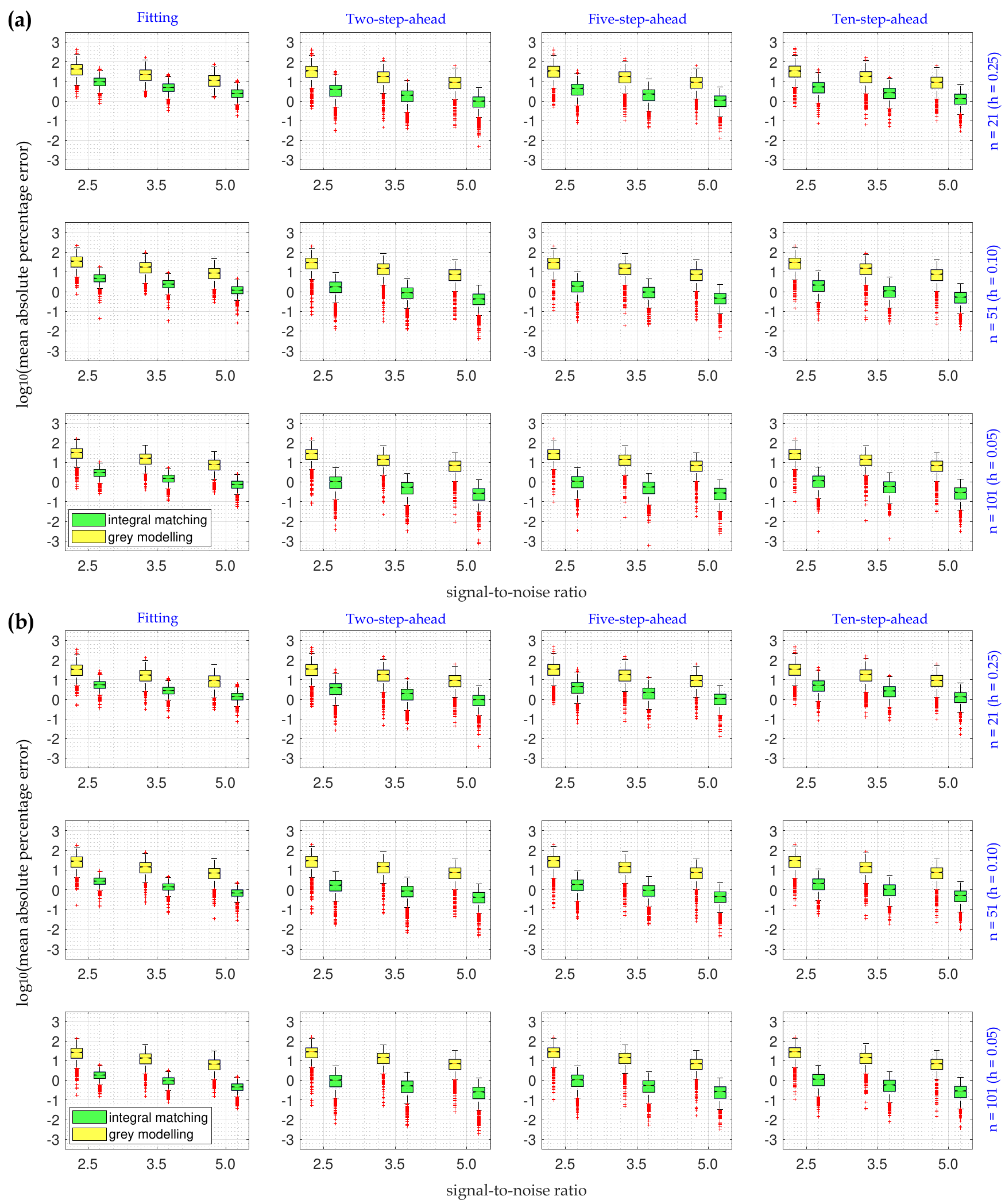}
    \caption{Boxplots of the fitting and multi-step-ahead forecasting errors in 1000 replications: (a) the first component $x_1(t)$, (b) the second component $x_2(t)$.}
    \label{fig:5}
\end{figure}

\subsection{Evaluation of fitting and forecasting accuracy}\label{sec:5-3}

Since grey models are widely used to forecast time series, it is also necessary to analyse the distributions of the fitting and forecasting errors. Figure \ref{fig:5} depicts the boxplots of the fitting errors, two-step-ahead, five-step-ahead, and ten-step-ahead forecasting errors in 1000 simulations.

As a whole, for both the first component $x_1$ in Figure \ref{fig:5}(a) and the second component $x_2$ in Figure \ref{fig:5}(b), integral matching has much smaller fitting and multi-step-ahead forecasting errors than grey modelling in all the sample size and signal-to-noise ratio combination scenarios, indicating that the former outperforms the latter from both fitting and forecasting viewpoints.
Specifically, by the same token as the grouping manner in section \ref{sec:5-2},  we may also conclude the following:
\begin{enumerate}
    \item [(i)]
    For each group with an equal sample size, i.e.,  each row in Figure \ref{fig:5}(a) and \ref{fig:5}(b), the fitting and multi-step-ahead forecasting errors decrease successively with the increase of signal-to-noise ratio, indicating the impact of measurement noise on the forecasting accuracy. Let's take the example of the integral matching method. Letting $n=21$, then, as $snr$ increases from 2.5 through 3.5 to 5.0, the medians of fitting errors decreases from 6.34\% through 3.22\% to 1.58\% (for $x_1$) and from 11.57\% through 5.89\% to 2.89\% (for $x_2$), and meanwhile, the medians of the ten-step-ahead forecasting errors decreases from 6.63\% through 3.24\% to 1.59\% (for $x_1$) and from 6.49\% through 3.27\% to 1.61\% (for $x_2$).
    \item [(ii)]
    For each group with an equal signal-to-noise ratio, i.e., each column in Figure \ref{fig:5}(a) and \ref{fig:5}(b), as the sample size increases, the fitting and multi-step-ahead forecasting errors decrease monotonously, partially implying the asymptotic unbiasedness of forecasts. Likewise, we still use the integral matching method as an example. Letting $snr=2.5$, then, as $n$ increases from 21 through 51 to 101,
    the medians of fitting errors decreases from 6.34\% through 3.10\% to 2.02\% (for $x_1$)
    and from 11.57\% through 5.38\% to 3.44\% (for $x_2$), and meanwhile,
    the medians of the ten-step-ahead forecasting errors decreases from 6.43\% through 2.49\% to 1.38\% (for $x_1$)
    and from 6.49\% through 2.54\% to 1.42\% (for $x_2$).
\end{enumerate}

Overall, combining the conclusions in parameter estimation section \ref{sec:5-2}, we conclude that integral matching consistently outperforms grey modelling in terms of parameter estimation, fitting and multi-step-ahead forecasting, as well as robustness to noise disturbance.

\section{Real-world application}\label{sec:6}

\subsection{Data collection}

Water supply plays an important role in a country's stability and development, while groundwater and surface water are two main water supply sources.
Facing the global groundwater crisis \cite{famiglietti2014global}, the Chinese government and enterprises have invested in many water projects over the past decades (including the waste-water treatment and reuse, rain collection, seawater desalinization, and others) in order to partially replace the groundwater losses that have occurred previously. For simplicity of presentation, the water supplies from these water projects are collectively referred to as the `other water supply sources'.
Forecasting these other water supplies accurately is clearly important in efforts to realize sustainable water resource development and research in the future.
The total amount of other water supplies ($10^9 \mathsf{m}^3$) is collected from the National Bureau of Statistics of China (\url{http://data.stats.gov.cn/english/easyquery.htm?cn=C01}). However, water supply records are incomplete and sparse. Only limited data from 2004 to 2018 are available, as shown in Table \ref{tbl:a}.

\begin{table}[!ht]
    \centering
    \footnotesize
    \caption{The total amount of other water supplies in China from 2004 to 2018.}
    \label{tbl:a}
    \setlength{\tabcolsep}{1.7mm}
    \begin{tabular}{lc cc cc cc cc cc cc cc}
        \toprule
        Year & 2004 & 2005 & 2006 & 2007 & 2008 & 2009 & 2010 & 2011
             & 2012 & 2013 & 2014 & 2015 & 2016 & 2017 & 2018  \\
        Value & 17.20  & 21.96  & 22.70  & 25.70  & 28.74  & 31.16  & 33.12  & 44.80
              & 44.60  & 49.94  & 57.46  & 64.50  & 70.85  & 81.20  & 86.40 \\
        \bottomrule
    \end{tabular}
\end{table}

Using various estimated grey models, 5-step-ahead forecasts are generated in order to be consistent with the Five Year Plan of China, which sets the government's development goals for the next five-year stage. The data are divided into two parts: the data from 2004 to 2015 (12 samples taking 80\% of the 15 samples) are used as a training dataset to construct the models and the data from 2016 to 2018 (3 samples taking 20\% of the 15 samples) are used as test dataset to validate these identified models. In addition, the forecasts in 2019 and 2020 are used to evaluate the post-2018 behavior.

\subsection{Results of integral matching-based and grey models}

For simplicity of presentation, the integral matching-based differential equation model is abbreviated as IMDE. Table \ref{tbl:5} shows the results of IMDE models and the only difference among these models lies in the forcing terms which are expressed by the linear combination of the time polynomial basis functions. The results show the following:
\begin{enumerate}
    \item [(i)]
    In comparison with IMDE$_5$, IMDE$_1$--IMDE$_4$ have lower fitting errors (all the $\mathrm{MAPE}_\mathrm{in}$s less than 5.0\%) and forecasting errors (all the $\mathrm{MAPE}_\mathrm{out}$s less than 5.0\%), and meanwhile, IMDE$_5$ has more complex structure, indicating that IMDE$_1$--IMDE$_4$ are superior to IMDE$_5$.
    \item [(ii)]
    There are only small differences in the fitting error $\mathrm{MAPE}_\mathrm{in}$s among IMDE$_1$--IMDE$_4$, whereas IMDE$_1$, IMDE$_3$ and IMDE$_4$ have much lower forecasting errors in comparison with IMDE$_2$, indicating that IMDE$_2$ is inferior to the other three.
    \item [(iii)]
    The fitting and forecasting errors of IMDE$_3$ and IMDE$_4$ are equal to each other, respectively, whereas IMDE$_4$ has an additional term $b_2t^2$, indicating that IMDE$_3$ is more parsimonious (parametrically efficient) and so better according to the Occam's razor principle \cite{young1996simplicity}. In fact, the estimated parameters of IMDE$_4$ are $\hat{a}=-0.0395$, $\hat{c}=0.4509$, $\hat{b}_1=0.7717$, $\hat{b}_2=-0.0018$, $\hat{\eta}=20.9025$ where the estimates of $\hat{b}_2$ is close to 0; and the other estimates are close to those of IMDE$_3$ with $\hat{a}=-0.0458$, $\hat{c}=0.5761$, $\hat{b}_1=0.7730$, $\hat{\eta}=20.8931$.
    \item [(iv)]
    IMDE$_3$ has a higher fitting error $\mathrm{MAPE}_\mathrm{in}$ but lower forecasting error $\mathrm{MAPE}_\mathrm{out}$ than IMDE$_1$. The reason for the former is that the estimated initial value of IMDE$_3$ is much greater than the true value. Without consideration of the first sample, the $\mathrm{MAPE}_\mathrm{in}$s of IMDE$_1$ and IMDE$_3$ in the period from 2005 to 2015 are 3.24\% and 2.85\%, respectively, indicating that IMDE$_3$ is superior in terms of fitting and forecasting accuracies.
    In fact, the estimated parameters of IMDE$_1$ are $\hat{a}=0.1144$ and $\hat{\eta}=18.2176$, resulting in a time response function expressed as $x(t)=16.2482\exp(0.1144t)$. This increases too fast to be used in medium- and long-term prediction (the $\mathrm{APE}=4.61\%$ in 2018 also tends to confirm this).
\end{enumerate}

Overall, IMDE$_3$ is the best among the five IMDE models. Besides, corresponding to the model structure of IMDE$_3$, we consider the second-order grey polynomial model \cite{wei2018optimal}, abbreviated as GPM(1,1,2) in Table \ref{tbl:5}, and analyse the similarities and differences, as shown in Table \ref{tbl:5} and Figure \ref{fig:7}(a). Since the fitting values in 2004 are actually the estimated initial values, the differences between initial values are the main reason for the differences in performance. The details are present in the following:

The expressions of IMDE$_3$ and GPM(1,1,2) models are
\[
    \frac{d}{dt}x_\mathrm{m}(t)=-0.04578x_\mathrm{m}(t)+0.7730t+0.5761,~ x(1)=20.8931,~ t\geq 1
\]
and
\[
    \frac{d}{dt}y_\mathrm{g}(t)=-0.04578y_\mathrm{g}(t)+0.3865t^2+0.9626t+20.6123,~ y(1)=21.5509,~ t\geq 1,
\]
respectively; then, the time response functions are
\[
    x_\mathrm{m}(t)=16.8847 t + 377.1157\exp(-0.04578t) - 356.2318,~ t\geq 1
\]
and
\[
    y_\mathrm{g}(t) = 8.4424t^2 - 347.7895 t - 8046.2287\exp(-0.04578 t)  + 8047.0682,~ t\geq 1,
\]
which leads to the restored time response function expressed as
\[
    x_\mathrm{g}(t)=\frac{y_\mathrm{g}(t)-y_\mathrm{g}(t-1)}{t-(t-1)}
        = 16.8847t + 376.9255\exp(-0.04578t) - 356.2318,~ t\geq 2.
\]

By comparing $x_\mathrm{m}(t)$ with $x_\mathrm{g}(t)$, IMDE$_3$ and GPM(1,1,2) models have the same time response functions except for the coefficients of the exponential terms which are caused by the estimated initial values. Overall, IMDE$_3$ has simpler modelling procedures and is easier to explain, so it is probably the best one to use in this application.

\begin{table}[!ht]
    \centering
    \footnotesize
    \begin{threeparttable}
        \caption{Fitting and forecasting results obtained from integral matching-based and grey models.}
        \label{tbl:5}
        \setlength{\tabcolsep}{2.6mm}
        \begin{tabular}{lr rr rr rr rr rr r}
        \toprule
        Year & \multicolumn{2}{l}{IMDE$_1$ \textsuperscript{i}} & \multicolumn{2}{l}{IMDE$_2$ \textsuperscript{ii}} &
           \multicolumn{2}{l}{IMDE$_3$ \textsuperscript{iii}} & \multicolumn{2}{l}{IMDE$_4$ \textsuperscript{iv}} &
           \multicolumn{2}{l}{IMDE$_5$ \textsuperscript{v}} & \multicolumn{2}{l}{GPM(1,1,2) \textsuperscript{vi}} \\
           \cmidrule(lr){2-3} \cmidrule(lr){4-5} \cmidrule(lr){6-7} \cmidrule(lr){8-9} \cmidrule(lr){10-11} \cmidrule(lr){12-13}
 & Value & APE & Value & APE & Value & APE & Value & APE & Value & APE & Value & APE \\ \hline
2004 & 18.22  & 5.92  & 19.14  & 11.27  & 20.89  & 21.47  & 20.90  & 21.53  & 22.98  & 33.60  & 21.55  & 25.30 \\
2005 & 20.43  & 6.99  & 21.12  & 3.84  & 21.66  & 1.38  & 21.67  & 1.33  & 22.32  & 1.66  & 21.48  & 2.17 \\
2006 & 22.90  & 0.89  & 23.37  & 2.97  & 23.14  & 1.95  & 23.15  & 2.00  & 23.28  & 2.57  & 22.98  & 1.22 \\
2007 & 25.68  & 0.09  & 25.95  & 0.97  & 25.32  & 1.49  & 25.33  & 1.45  & 25.47  & 0.90  & 25.16  & 2.10 \\
2008 & 28.79  & 0.17  & 28.89  & 0.52  & 28.15  & 2.05  & 28.16  & 2.02  & 28.55  & 0.64  & 28.00  & 2.57 \\
2009 & 32.28  & 3.59  & 32.24  & 3.47  & 31.61  & 1.45  & 31.62  & 1.47  & 32.29  & 3.63  & 31.47  & 0.99 \\
2010 & 36.19  & 9.27  & 36.07  & 8.90  & 35.68  & 7.72  & 35.68  & 7.73  & 36.50  & 10.21  & 35.54  & 7.30 \\
2011 & 40.58  & 9.43  & 40.43  & 9.74  & 40.31  & 10.02  & 40.32  & 10.01  & 41.10  & 8.26  & 40.18  & 10.31 \\
2012 & 45.49  & 2.01  & 45.42  & 1.83  & 45.50  & 2.01  & 45.50  & 2.02  & 46.09  & 3.35  & 45.37  & 1.73 \\
2013 & 51.01  & 2.14  & 51.10  & 2.33  & 51.20  & 2.53  & 51.21  & 2.54  & 51.60  & 3.32  & 51.08  & 2.29 \\
2014 & 57.19  & 0.47  & 57.59  & 0.22  & 57.41  & 0.08  & 57.42  & 0.08  & 57.86  & 0.69  & 57.30  & 0.28 \\
2015 & 64.12  & 0.58  & 64.99  & 0.76  & 64.10  & 0.62  & 64.10  & 0.62  & 65.23  & 1.14  & 63.99  & 0.79 \\
\multicolumn{2}{l}{MAPE$_\mathrm{in}$ (\%)} & 3.46  &  & 3.90  &  & 4.40  &  & 4.40  &  & 5.83  &  & 4.75 \\
2016 & 71.90  & 1.47  & 73.43  & 3.65  & 71.24  & 0.55  & 71.24  & 0.55  & 74.25  & 4.80  & 71.14  & 0.40 \\
2017 & 80.61  & 0.73  & 83.07  & 2.30  & 78.82  & 2.93  & 78.81  & 2.94  & 85.60  & 5.42  & 78.72  & 3.06 \\
2018 & 90.38  & 4.61  & 94.07  & 8.87  & 86.81  & 0.48  & 86.80  & 0.46  & 100.15  & 15.91  & 86.72  & 0.37 \\
\multicolumn{2}{l}{MAPE$_\mathrm{out}$ (\%)} & 2.27  &  & 4.94  &  & \textbf{1.32}  &  & \textbf{1.32}  &  & 8.71  &  & \textbf{1.28} \\
2019 & 101.33  &      & 106.61  &      & 95.21  &       & 95.18  &       & 118.98  &        & 95.12   &   \\
2020 & 113.61  &      & 120.93  &      & 103.98  &      & 103.93  &      & 143.40  &        & 103.89  &   \\
            \bottomrule
        \end{tabular}
        \begin{tablenotes}
            \item Note that the model structures are (i) $\frac{d}{dt}x(t)=ax(t)$,
            (ii) $\frac{d}{dt}x(t)=ax(t)+c$,
            (iii) $\frac{d}{dt}x(t)=ax(t)+b_1t+c$,
            (iv) $\frac{d}{dt}x(t)=ax(t)+b_1t+b_2t^2+c$,
            (v) $\frac{d}{dt}x(t)=ax(t)+b_1t+b_2t^2+b_3t^3+c$,
             and (vi) $\frac{d}{dt}y(t)=ay(t)+b_1t+b_2t^2+c$.
        \end{tablenotes}
    \end{threeparttable}
\end{table}

\begin{figure}[!ht]
    \centering
    \includegraphics[scale=0.72, trim = 0 0 0 0, clip = true]{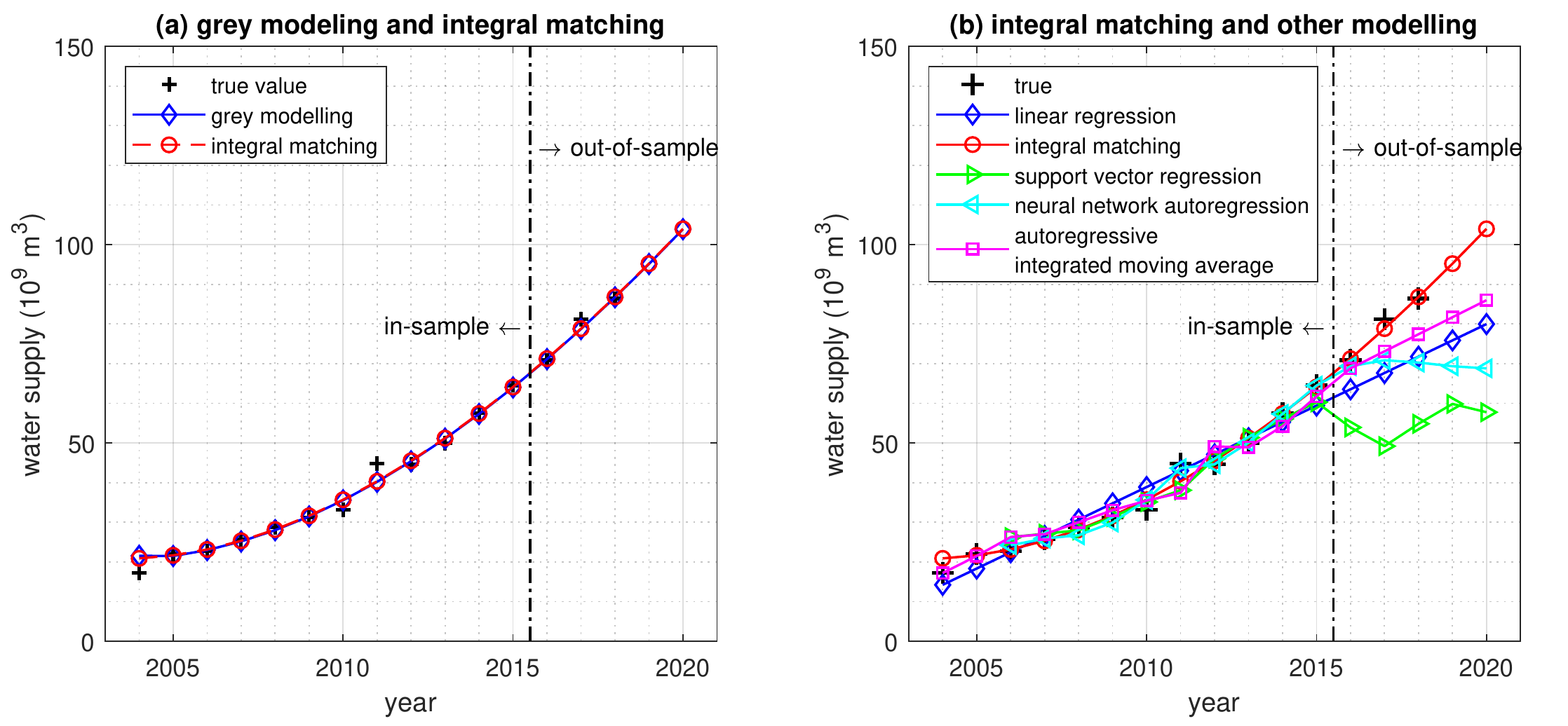}
    \caption{Fitting and forecasting results of other water supply obtained from various models.}
    \label{fig:7}
\end{figure}

\subsection{Comparison with other classical models}

The integral matching-based model (IMDE$_3$) is compared with the linear regression (LR), autoregressive integrated moving average (ARIMA), neural network autoregression (NNAR) and support vector regression (SVR). The calculation results associated with the absolute percentage errors (APEs) at each time instant are shown in Table \ref{tbl:6} and Figure \ref{fig:7}(b). ARIMA and NNAR models are implemented by the $\mathsf{auto.arima}$ and $\mathsf{nnetar}$ functions in $\mathsf{forecast}$ package \cite{hyndman2008automatic}, and SVR model by the $\mathsf{svm}$ function in $\mathsf{e1071}$ package \cite{dimitriadou2019package} in $\mathsf{R}$ software.

\begin{table}[!ht]
    \centering
    \footnotesize
    \begin{threeparttable}
        \caption{Fitting and forecasting results obtained from other classical models.}
        \label{tbl:6}
        \setlength{\tabcolsep}{3.9mm}
        \begin{tabular}{lr rr rr rr rr rr r}
            \toprule
            Year & \multicolumn{2}{l}{IMDE$_3$ \textsuperscript{i}} & \multicolumn{2}{l}{LR \textsuperscript{ii}} &
            \multicolumn{2}{l}{ARIMA \textsuperscript{iii}} & \multicolumn{2}{l}{NNAR \textsuperscript{iv}} &
            \multicolumn{2}{l}{SVR \textsuperscript{v}} \\
            \cmidrule(lr){2-3} \cmidrule(lr){4-5} \cmidrule(lr){6-7} \cmidrule(lr){8-9} \cmidrule(lr){10-11}
            & Value & APE & Value & APE & Value & APE & Value & APE & Value & APE \\ \hline
            2004 & 20.89  & 21.47  & 14.22  & 17.30  & 17.19  & 0.07  &  &  &  &  \\
            2005 & 21.66  & 1.38  & 18.33  & 16.51  & 21.50  & 2.09  &  &  &  &  \\
            2006 & 23.14  & 1.95  & 22.44  & 1.13  & 26.26  & 15.68  & 24.21  & 6.66  & 26.26  & 15.70 \\
            2007 & 25.32  & 1.49  & 26.55  & 3.31  & 27.00  & 5.06  & 25.88  & 0.71  & 27.11  & 5.47 \\
            2008 & 28.15  & 2.05  & 30.66  & 6.68  & 30.00  & 4.38  & 26.74  & 6.96  & 27.98  & 2.64 \\
            2009 & 31.61  & 1.45  & 34.77  & 11.58  & 33.04  & 6.03  & 29.93  & 3.96  & 31.33  & 0.55 \\
            2010 & 35.68  & 7.72  & 38.88  & 17.38  & 35.46  & 7.07  & 35.63  & 7.57  & 35.11  & 6.00 \\
            2011 & 40.31  & 10.02  & 42.99  & 4.05  & 37.42  & 16.47  & 43.64  & 2.60  & 38.07  & 15.02 \\
            2012 & 45.50  & 2.01  & 47.10  & 5.59  & 49.10  & 10.09  & 44.63  & 0.07  & 46.02  & 3.18 \\
            2013 & 51.20  & 2.53  & 51.20  & 2.53  & 48.90  & 2.08  & 50.11  & 0.34  & 51.36  & 2.84 \\
            2014 & 57.41  & 0.08  & 55.31  & 3.74  & 54.24  & 5.60  & 57.32  & 0.24  & 56.05  & 2.46 \\
            2015 & 64.10  & 0.62  & 59.42  & 7.87  & 61.76  & 4.25  & 64.46  & 0.05  & 60.25  & 6.59 \\
            \multicolumn{2}{l}{MAPE$_\mathrm{in}$ (\%)} & 4.40  &  & 8.14  &  & 6.57  &  & 2.92  &  & 6.05 \\
            2016 & 71.24  & 0.55  & 63.53  & 10.33  & 68.80  & 2.89  & 69.34  & 2.13  & 53.91  & 23.91 \\
            2017 & 78.82  & 2.93  & 67.64  & 16.70  & 73.10  & 9.98  & 70.85  & 12.74  & 49.19  & 39.42 \\
            2018 & 86.81  & 0.48  & 71.75  & 16.96  & 77.40  & 10.42  & 70.27  & 18.67  & 54.83  & 36.54 \\
            \multicolumn{2}{l}{MAPE$_\mathrm{out}$ (\%)} & \textbf{1.32}  &  & 14.66  &  & 7.76  &  & 11.18  &  & 33.29 \\
            2019 & 95.21  &  & 75.86  &  & 81.70  &  & 69.34  &  & 59.84  &  \\
            2020 & 103.98  &  & 79.97  &  & 86.00  &  & 68.80  &  & 57.75  &  \\
            \bottomrule
        \end{tabular}
        \begin{tablenotes}
            \item Note that (i) IMDE$_3$ is the copy of that in Table \ref{tbl:5} for easier comparison; (ii) the coefficient of determination is 0.9548; (iii) the model order is arima(0,1,0); (iv) the model type is feed-forward neural network with 2 lagged inputs and 3 nodes in the only hidden layer; (v) the model type is $\epsilon$-SVR with the embedding dimension equal to 2 and the kernel type being radial basis function.
        \end{tablenotes}
    \end{threeparttable}
\end{table}

Table \ref{tbl:6} shows that all the MAPE$_\mathrm{in}$s of the five models are less than 10\% but their MAPE$_\mathrm{out}$s vary widely, indicating that they have almost the same fitting but different forecasting performance.
The high coefficient of determination of LR indicates that the model order is appropriate, but it does not consider the autocorrelation in time series and thus results in poor forecasting performance. On the contrary, ARIMA consider the autocorrelation, thereby having a better performance than LR. NNAR obtains the smallest fitting error but a higher forecasting error, and SVR performs better in fitting but behaves poorly in forecasting. The reason for good fitting but bad forecasting performance of NNAR and SVR may be that the small modelling sample size leads to the over-fitting, although we have tried our best to avoid this by simplifying the model structures in the modelling process.

Overall, the integral matching-based model, IMDE$_3$, has acceptable fitting error (MAPE$_\mathrm{in}=4.40\%$) and minimal forecasting error (MAPE$_\mathrm{out}=1.32\%$) and so can be considered the best in this example (see also Figure \ref{fig:7}(b)).

\section{Conclusions}\label{sec:7}


This paper has presented a unified framework for continuous-time grey models that subsume the single-variable, multi-variable, and multi-output grey models as special cases. It has shown that this unified form can be reduced into an equivalent differential equation and how the integral operator links the differential equations of grey and reduced models. It has also shown that structural parameters and initial condition of the reduced model can be simultaneously estimated by the integral matching approach. We see, therefore, that grey and reduced models are both concerned with the data-based modelling in terms of differential equations. The former uses the Cusum operator explicitly, while the latter implicitly employs the Cusum operator via a numerical discretization scheme of the integral operator. As pointed out in Section \ref{sec:4-2}, in practical applications, we recommend the integral matching-based reduction form due to its advantages over grey models, including higher accuracy, simpler modelling procedures and easier explanation of modelling results.


This paper has not described the impact of observed errors on the estimates of grey models and reduced models. The presence of errors-in-variables (also known as measurement error in statistics) induces an asymptotic bias on the parameter estimates which is a function of the signal-to-noise ratio on ${\bm{x}}(t_k)$ and is zero only when there is no noise. Other methods of estimation have advantages in this regard, such as the refined instrumental variable estimation; another possibility is two-stage linear least squares where the original time-series observations are de-noised at the first stage and then the parameters are estimated using linear least squares at the second stage.
Besides, both grey and reduced models use linear least squares-based statistical estimation but do not consider the statistical properties of their corresponding estimators. Consequently, future research should be focused on the statistical properties concerning minimum variance and asymptotic analysis, that is, the asymptotic consistency and efficiency respectively measuring the convergence and variance of the estimators as the time interval $h$ tends to 0.

\section*{Data availability}

Data and codes could be seen in the following link: \url{https://github.com/weibl9/CTGMs}.

\section*{Acknowledgements}

This work was supported by the National Natural Science Foundation of China (72171116, 71671090) and the Fundamental Research Funds for the Central Universities of China (NP2020022).

This research was completed during the first author's visit at Lancaster University; thanks to Wlodek Tych for hosting my visit. The authors are most grateful to Peter C. Young for his valuable suggestions. This does not mean that he necessarily shares the views stated in the paper, and the authors, of course, responsible for any errors or omissions.

\bibliographystyle{model1-num-names}
\bibliography{ctgmBib}

\end{document}